\documentclass[conference]{IEEEtran}

\usepackage{amsmath}
\usepackage{amssymb}
\usepackage{amsthm}
\usepackage{mathtools}
\usepackage{xspace}
\usepackage{tikz}\usetikzlibrary{arrows,automata,positioning,shapes}
\usepackage[T1]{fontenc}
\usepackage[basic,bold]{complexity}
\usepackage{thmtools}
\usepackage{enumitem}
\usepackage{stmaryrd}
\usepackage{nicefrac}

\usepackage{thmtools, thm-restate}

\newclass{\NTWOEXP}{N2EXP}
\newclass{\SIGTWO}{\Sigma^{\P}_2}

\renewcommand{\epsilon}{\varepsilon}

\newcommand*{\eg}{e.g.\@\xspace}
\newcommand*{\ie}{i.e.\@\xspace}

\newcommand{\N}{\mathbb{N}}
\newcommand{\Z}{\mathbb{Z}}
\newcommand{\Q}{\mathbb{Q}}
\newcommand{\Qpos}{\mathbb{Q}_{\ge 0}}
\newcommand{\ZeroOne}{(0, 1]}

\newcommand{\defeq}{\coloneqq}

\newcommand{\bI}{\mathbf{I}}

\newcommand{\bX}{\mathbf{X}}
\newcommand{\bY}{\mathbf{Y}}
\newcommand{\bZ}{\mathbf{Z}}
\newcommand{\bT}{\mathbf{T}}

\newcommand{\Oh}{\mathcal{O}}
\newcommand{\V}{\mathcal{V}}
\renewcommand{\W}{\mathcal{W}}
\newcommand{\Pp}{\mathcal{P}}

\theoremstyle{plain}
\newtheorem{theorem}{Theorem}
\newtheorem{lemma}[theorem]{Lemma}
\newtheorem{proposition}[theorem]{Proposition}
\newtheorem{corollary}[theorem]{Corollary}

\theoremstyle{definition}

\theoremstyle{remark}

\newcommand{\reach}{\mathrm{Reach}}
\newcommand{\steps}[1]{\rightarrow_{#1}}
\newcommand{\cl}[1]{\overline{#1}}

\newcommand{\set}[1]{\{ #1 \}}
\newcommand{\intervals}{\mathcal{I}}
\renewcommand{\int}{\mathrm{int}}
\newcommand{\encodings}{\mathcal{E}}
\newcommand{\minkowski}{\mathrm{msum}}

\def \ifempty#1{\def\temp{#1} \ifx\temp\empty }

\newcommand{\Post}[2]{\ifempty{#2} \mathrm{Post}_{#1} \else \mathrm{Post}_{#1}(#2) \fi}

\newcommand{\enab}[1]{\mathrm{enab}(#1)}
\newcommand{\In}[1]{\mathrm{in}(#1)}
\newcommand{\Out}[1]{\mathrm{out}(#1)}
\newcommand{\Effect}[1]{\Delta(#1)}
\newcommand{\effect}[1]{\Effect{#1}}
\newcommand{\Effectp}[1]{\Delta^+(#1)}
\newcommand{\Effectn}[1]{\Delta^-(#1)}
\newcommand{\TheBetterPost}{\mathrm{Succ}}

\newcommand{\Acceleration}{\mathrm{Acc}}

\newcommand{\bigO}{\mathcal{O}}

\newcommand{\guard}[1]{\tau(#1)}

\newcommand{\act}{\mathrm{active}}

\DeclarePairedDelimiter\abs{\lvert}{\rvert}%

\newcommand{\FromTo}[3]{#1[#2..#3]}

\newcommand*{\myproofname}{Proof of the claim}

\newcommand{\foq}{$\mathrm{FO}(\mathbb{Q}, +,<)$\xspace}

\usepackage{todonotes}
\usepackage[hidelinks]{hyperref}
\hypersetup{
  colorlinks   = true,
  urlcolor     = blue,
  linkcolor    = blue,
  citecolor    = red
}

\begin{document}

\onecolumn

% \bstctlcite{IEEEexample:BSTcontrol}

\title{Continuous One-Counter Automata}

\author{\IEEEauthorblockN{Michael Blondin\IEEEauthorrefmark{1},
Tim Leys\IEEEauthorrefmark{2}
Filip Mazowiecki\IEEEauthorrefmark{3},
Philip Offtermatt\IEEEauthorrefmark{1}\IEEEauthorrefmark{3}, and
Guillermo A. P\'erez\IEEEauthorrefmark{2}}
\IEEEauthorblockA{\IEEEauthorrefmark{1}Universit\'e de Sherbrooke, Canada}
\IEEEauthorblockA{\IEEEauthorrefmark{2}University of Antwerp, Belgium}
\IEEEauthorblockA{\IEEEauthorrefmark{3}Max Planck Institute for Software
  Systems, Germany}}

\maketitle

\begin{abstract}
  We study the reachability problem for continuous
  one-counter automata, COCA for short. In such automata, transitions
  are guarded by upper and lower bound tests against the counter
  value. Additionally, the counter updates associated with taking
  transitions can be (non-deterministically) scaled down by a nonzero
  factor between zero and one. Our three main results are as follows:
  (1)~We prove that the reachability problem for COCA with global
  upper and lower bound tests is in NC2; (2)~that, in general, the
  problem is decidable in polynomial time; and (3)~that it is
  decidable in the polynomial hierarchy for COCA with parametric
  counter updates and bound tests.
\end{abstract}


% \fi
\IEEEpeerreviewmaketitle

\section{Introduction}
\label{sec:introduction}
Counter machines form a fundamental computational model which captures the
behavior of infinite-state systems. Unfortunately, their central decision
problems, like the \emph{(configuration) reachability problem}, are
undecidable as the model is Turing-complete~\cite{Minsky61,Minsky67}. To
circumvent this issue, numerous restrictions of counter machines have been studied in the literature. For
instance, vector addition systems with states (VASS) arise from restricting
the type of tests that can be used to guard
transitions~\cite{Lipton76,Mayr81,LerouxS19,clllm19}. One-counter automata are yet another
well-studied model~\cite{hkow09,fj15,bqs19}, in this case arising from the restriction to
a single counter, hence the name.
%In this work, we focus on this second restriction.

We consider one-counter automata that can use tests of the form ``$\leq c$''
and ``$\geq d$'' --- where $c$ and $d$ are constants --- to guard their
transitions. As a natural extension of finite-state automata, one-counter
automata allow for better conservative approximations of classical
static-analysis problems like instruction reachability (see, \eg,
\emph{program graphs} as defined in~\cite{bk08}). They also enable the
verification of programs with lists~\cite{bbhimv11} and XML-stream
validation~\cite{cr04}.  Furthermore, their reachability problem seems
intrinsically connected to that~of timed automata (TA). The reachability
problem for two-clock TA is known to be logspace-equivalent to the same
problem for succinct one-counter automata (SOCA), that is, where constants used
in counter updates and tests are encoded in binary~\cite{how12}. An analogue
of this connection holds when SOCA are enriched with parameters that can be
used on updates: reachability for two-parametric-clock TA reduces to the
(existential) reachability problem for parametric one-counter
automata~\cite{bo17}. Interestingly, Alur et al.\ observe~\cite{ahv93} the
former subsumes a long-standing open problem of Ibarra~\cite{ijtw93}
concerning ``simple programs''.

All of the above connections from interesting problems to reachability for
SOCA and parametric SOCA indicate that efficient algorithms for the problem
are very much desirable. Unfortunately, it is known that reachability for SOCA
(with upper and lower-bound tests) is \PSPACE-complete~\cite{fj15}. For
parametric SOCA, the situation is even worse as the general problem is not
even known to be decidable. In this work, we study \emph{continuous
relaxations} of these problems and show that their complexities belong in
tractable complexity classes.  We thus give the first efficient
conservative approximation for the reachability problem for SOCA and parametric
SOCA.

\paragraph*{The continuous relaxations}
We observe that the model considered by
Fearnley and Jurdzi\'{n}ski~\cite{fj15} is not precisely our SOCA, rather they consider
\emph{bounded $1$-dimensional VASS}. In such VASS, the counter is not allowed
to take negative values. Additionally, it is not allowed to take values
greater than some global upper bound.  Note that inequality tests against
constants can be added to such VASS as ``syntactic sugar'' since these can be
implemented making use of the upper and lower bounds. These observations allow
us to adapt Blondin and Haase's definition of continuous VASS~\cite{bh17}
to introduce \emph{(bounded) continuous one-counter automata} (COCA) which have
\emph{global} upper and lower-bound tests: Transitions are allowed to be
``partially taken'' in the sense that the respective counter updates can be
scaled by some factor $\alpha \in (0,1]$.

In contrast to the situation in the discrete world, because of the continuous
semantics, adding arbitrary upper and lower-bound tests to COCA does
result in the more expressive model of \emph{guarded COCA}.  Importantly,
guarded COCA are a ``tighter'' relaxation of SOCA than COCA
are (via the translation to bounded $1$-VASS).
Finally, we also study the reachability problem for \emph{parametric COCA}.
These are guarded COCA where counter updates can be variables $x \in X$ whose
values range over the rationals; bound tests can also be against variables
from $X$.  The resulting model can be seen as a continuous relaxation of
Ibarra's simple programs~\cite{ijtw93,bo17}.

\paragraph*{Contributions}

Our main contributions are three-fold (see \autoref{main_results}). First, we
show that the reachability problem for COCA is decidable in $\NC^2$.
Second, we give a polynomial-time algorithm for the same problem for guarded
COCA. Finally, we show that the reachability problem for parametric COCA belongs to
$\SIGTWO$ and is \NP-hard.

On the way, we prove that the
reachability problem for COCA enriched with equality tests is in
$\NC^2$; that the reachability problem for parametric COCA where only counter
updates are allowed to be parametric is equivalent to the integer-valuation
restriction of the problem; and that the reachability problem for acyclic
parametric COCA is \NP-complete.

\paragraph*{Other related work}
To complete a full circle of connections between timed and counter
automata, we note that the closest model to ours is that of one-clock
TA. The value of the clock in such automata evolves (continuously)
at a fixed positive rate and can be reset by some transitions. Guarded COCA
can simulate clock delays using $+1$ self-loops and resets using $-1$
self-loops and bound tests ``$\leq 0$'' and ``$\geq 0$''. Our model thus generalises
one-clock TA.

The reachability problem for (non-parametric) one-clock TA is
\NL-complete~\cite{LMS04}. The \NL{} membership proof from~\cite{LMS04} relies on
the fact that clock delays can always occur and do so
without changing the state. This does not hold in the more general framework of COCA. Consequently, the proof does not extend directly to COCA.

The reachability problem for parametric one-clock TA with integer-valued
parameters is known to be decidable in \NEXP{}~\cite{bo17}. Since
non-parametric clocks can be removed at the cost of an exponential
blow-up~\cite{ahv93}, it is also argued in~\cite{bo17} that the problem belongs to \NTWOEXP{} if an arbitrary number of non-parametric clocks is
allowed~\cite{bo17}. For the latter problem, the authors also prove that it is
\NEXP-hard. Our $\SIGTWO$ upper bound for update-parametric COCA with
integer-valued parameters improves the latter two bounds.

\section{Preliminaries}
\label{sec:preliminaries}
We write $\Qpos$ for the set of nonnegative rationals, and $\Q_{>0}$
for the set of positive rationals.  We use symbols ``$[$'' and ``$]$''
for closed intervals, and ``$($'' and ``$)$'' for open intervals of
rational numbers. For example, $[a, b)$ denotes $\set{q \in \Q \mid
a \le q < b}$. Intervals do not have to be bounded, \eg\ we allow
$[3,+\infty)$. We denote the set of all intervals over $\Q$ by
$\intervals$. We write $\cl{X}$ to denote the \emph{closure} of a set
$X \subseteq \Q$, \ie\ $X$ enlarged with its limit points. For
example, $\cl{(3, 5)} = [3, 5]$, $\cl{[1, 4) \cup (4, 5]} = [1, 5]$
and $\cl{(-\infty, +\infty)} = (-\infty, +\infty)$. Throughout the
paper, numbers are encoded in \emph{binary} and we assume intervals to
be encoded as pairs of endpoints, together with binary flags
indicating whether the endpoints are contained or not.

    %% Since intervals are infinite, we encode them by writing down their endpoints as follows. Let $\invbreve{\Q} = \{\invbreve{a} \mid a \in \Q\}$ denote a disjoint copy of $\Q$. We represent $\intervals$ by $(\set{-\infty} \cup \invbreve{\Q} \cup \Q) \times (\Q \cup \invbreve{\Q} \cup \set{+\infty})$. The intended meaning is that $\invbreve{a}$ encodes an open interval. For example, the interval $[a,b)$ is encoded by the pair $[a, \invbreve{b}]$.

%% It will be convenient to add elements in $\Q$ and $\invbreve{\Q}$. Fix $a,b,c \in \Q$ such that $a+b = c$. Then we define $a + \invbreve{b} = \invbreve{a} + b = \invbreve{c}$. We also define $\invbreve{a} + \infty = +\infty + \invbreve{a} = +\infty$ and similarly with $-\infty$. We will never add two elements from $\invbreve{Q}$, hence such operations are undefined. Sometimes to avoid case analyses we will write $a \le \invbreve{b}$ meaning $a < b$ and $a \ge \invbreve{b}$ meaning $a > b$.

%\michael{Notation tests: $\invbreve{a}$, $\mathring{a}$, $\cl{a}$}
%\guillermo{$\mathring{a}$ is not used though, right? the others I like}
%, $[\overset{\multimapinv}{3}, \overset{\multimap}{4}]$, $[3_{+}, 4_{-}]$

\subsection{One-counter automata}

A \emph{continuous one-counter automaton} (COCA) is a triple $\V =
(Q,T,\tau)$, where $Q$ and $T \subseteq Q \times \Z \times Q$ are finite
sets of \emph{states} and \emph{transitions}, and $\tau \in \intervals$. A \emph{configuration}
of $\V$ is a pair $(q,a) \in Q \times \Q$, denoted $q(a)$. A
\emph{run} from $p(a)$ to $q(b)$ in $\V$ is a sequence $\alpha_1 t_1
\cdots \alpha_n t_n$, where $\alpha_i \in (0,1]$ and $t_i = (q_{i-1},
z_i, q_i) \in T$, for which there exist configurations $q_0(a_0),
\ldots, q_n(a_n)$ such that $q_0(a_0) = p(a)$, $q_n(a_n) = q(b)$ and
$a_i = a_{i-1} + \alpha_i \cdot z_i$ for all $i \in
\set{1,\ldots,n}$. We say that such a run is \emph{admissible} if
$a_0, \ldots, a_n \in \tau$. 
For readers familiar with one-counter automata, note that the
model of
one-counter nets is obtained by setting
%the standard
%(discrete) model of bounded one-counter automata amounts to having
$\tau =
[0,+\infty)$ and $\alpha_i = 1$ for all $i$.
%\footnote{Technically, this would
%be a one-counter net as there are no zero-tests.}

A \emph{guarded COCA} is a triple $\W = (Q, T, \tau)$, where $(Q,T)$
is as for a COCA and $\tau : Q \to \intervals$ assigns intervals to states.
\emph{Configurations and runs} of $\W$ are defined as for a COCA.
A run
is \emph{admissible} if each of
its configurations $q_i(a_i)$ satisfies $a_i \in \tau(q_i)$. Hence, a COCA can
be seen as a guarded COCA where $\tau(q)$ is the same for all $q \in Q$.
%assigns interval $[0,+\infty)$
%to every state.
%For every state $q$, we will write $\topguard{q}$
%and $\botguard{q}$ to denote the interval $\guard{q}$ whose left and
%right endpoints are respectively relaxed to $-\infty$ and
%$+\infty$. For example, if $\guard{q} = [3, 5)$, then $\topguard{q}
%= (-\infty, 5)$ and $\botguard{q} = [3, +\infty)$. In particular,
%  $\guard{q} = \topguard{q} \cap \botguard{q}$.

%\michael{It seems like $\topguard{q}$ and $\botguard{q}$ have almost disappeared from the paper. Do we keep them here?}

The set $\intervals_X$ of \emph{parameterised intervals} over a set $X$
is the set of intervals whose endpoints belong either to $\Q \cup
\{-\infty, +\infty\}$ or $X$. A \emph{parametric COCA} is a tuple $\Pp
= (Q, T, \tau,X)$, where $Q$, $X$ and $T \subseteq Q \times (\Z \cup
X) \times Q$ are finite sets of \emph{states}, \emph{parameters} and
\emph{transitions}; and where $\tau \colon Q \to \intervals_X$. A
\emph{valuation} of $X$ is a function $\mu \colon X \to \Q$. We write
$\Pp^\mu = (Q,T^\mu,\tau^\mu)$ to denote the guarded COCA obtained
from $\Pp$ by replacing each parameter $x \in X$, occurring in $T$ and
$\tau$, with $\mu(x)$. We say that there is a run from $p(a)$ to
$q(b)$ in $\Pp$ if there exists a valuation $\mu$ such that $\Pp^\mu$
has a run from $p(a)$ to $q(b)$. In particular, $\Pp$ is a guarded
COCA if $X = \emptyset$. Otherwise, the notion of run only makes sense
w.r.t.\ a valuation $\mu$, \ie\ in the guarded COCA~$\Pp^\mu$.

In summary, we deal with three increasingly richer models: COCA
$\subseteq$ guarded COCA $\subseteq$ parametric COCA. In all variants,
the \emph{size} of the automaton is $|Q| + |T|\cdot s$, where $s$ is
the maximal number of bits required to encode a number in $T$ and
$\tau$.

\subsection{Runs, paths and cycles}

\newcommand{\SPaths}[1]{\mathrm{Paths}(#1)}
\newcommand{\Paths}[2]{\mathrm{Paths}_{#1,#2}}
\newcommand{\upath}[1]{\mathrm{path}(#1)}

Let $\W = (Q, T, \tau)$ be a guarded COCA. We write $\Paths{p}{q}$ to
denote the set of paths from state $p \in Q$ to state $q \in Q$ in the
graph induced by $T$. Let $\rho = \alpha_1 t_1 \cdots \alpha_n t_n$ be
a run where each $t_i = (q_{i-1}, z_i, q_i)$. The \emph{underlying
  path} of $\rho$ is $\upath{\rho} \defeq t_1 \cdots t_n \in
\Paths{q_0}{q_n}$. We further define $\FromTo{\rho}{i}{j} \defeq \alpha_i t_i
\cdots \alpha_j t_j$, $\rho_i \defeq \FromTo{\rho}{i}{i}$, $\In{\rho} \defeq
q_0$, $\Out{\rho} \defeq q_n$ and $\Effect{\rho} \defeq \sum_{i=1}^n
\alpha_i z_i$. By convention, $\FromTo{\rho}{i}{j} \defeq \varepsilon$ if $j <
i$, and $\Effect{\varepsilon} \defeq 0$. We write $p(a) \steps{\rho}
q(b)$ to denote the fact that $\rho$ is admissible from $p(a)$ to
$q(b)$. Since states $p$ and $q$ are determined by $\rho$, we may omit them
and simply write $a \steps{\rho} b$. For every $\beta \in (0, 1]$, we
define $\beta \rho \defeq (\beta \alpha_1) t_1 \cdots (\beta \alpha_n)
t_n$. Note that $\beta \rho$ is a run, but it may not preserve
admissibility.

%\guillermo{The delta of a transition is not yet defined, did we mean $1 t$ as
%in the scaling factor is $1$? or should we just move it after the following
%paragraph?}
%\michael{Yes, done.}

Let $\pi = t_1 \cdots t_n \in \Paths{p}{q}$ be such that each
$t_i = (q_{i-1}, z_i, q_i)$. We say that $\pi$ is a \emph{cycle} if
$p = q$, and \emph{simple} if $\pi$ does not repeat any state.
Let $\Effect{\pi} \defeq z_1 + \ldots + z_n$,
$\Effectp{\pi} \defeq \sum_{i = 1}^{n}
\max(0, z_i)$ and $\Effectn{\pi} \defeq \sum_{i = 1}^n
\min(0, z_i)$, with
$\Effect{\varepsilon} = \Effectp{\varepsilon} = \Effectn{\varepsilon}
\defeq 0$. In particular, $\Effectn{\pi} \leq
0 \leq \Effectp{\pi}$. Moreover, scaling the positive or negative
transitions of a path $\pi$ arbitrarily close to zero yields a run of
effect arbitrarily close to $\Effectn{\pi}$ or $\Effectp{\pi}$.

%\guillermo{The paragraph above had some intuition in nc2.tex before, should we
%say something here about why the definition will be useful?}
%\michael{Done.}

%\michael{Do we ever use simple cycles?}
%
%; and a \emph{simple cycle} if $\pi$ does not repeat any state except
%for $p = q$.

We write $p(a) \steps{\pi} q(b)$ to denote the existence of a run
$\rho$ such that $p(a) \steps{\rho} q(b)$ and $\upath{\rho} = \pi$. As
for runs, we may omit states and simply write $a \steps{\pi}
b$. The \emph{reachability function} given by $\pi$ is defined as
$\Post{\pi}{a} \defeq \set{b \in \Q \mid p(a) \steps{\pi} q(b)}$. We
generalise this notion to sets of paths and numbers:
\begin{align*}
  \Post{S}{A} &\defeq \bigcup_{\pi \in S} \bigcup_{a \in A} \Post{\pi}{a}.
\end{align*}
%In the 
%particular case where 
If $S = \Paths{p}{q}$, we write
$\Post{p,q}{a}$ and $\Post{p,q}{A}$. For example, for the guarded COCA
of \autoref{fig:coca:example}, the following holds:
$\mathrm{Post}_{p,q}(a) = (10, 18) \cup [19, 100)$ if $a = 15$; $(a-5,
a+3)$ if $a \in [-5, 15)$; and $\emptyset$ otherwise.

Finally, we define the set of starting points as
$\enab{\pi} \defeq \set{a \in \Q \mid \Post{\pi}{a} \neq \emptyset}$
and $\enab{S} \defeq \bigcup_{\pi \in S} \enab{\pi}$.

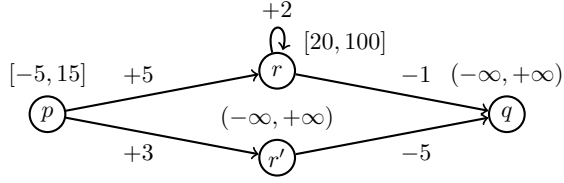
\begin{figure}
  \begin{center}
    \begin{tikzpicture}[auto, thick, transform shape, scale=0.9]
  \tikzstyle{astate} = [state, minimum size=15pt, inner sep=0pt];
  %% States
  \node[astate]                                  (p) {$p$};
  \node[astate, above right=0.25cm and 3cm of p] (r) {$r$};
  \node[astate, below right=0.25cm and 3cm of p] (s) {$r'$};
  \node[astate, below right=0.25cm and 3cm of r] (q) {$q$};

  %% Guards
  \node[above=0pt of p] {$[-5, 15]$};
  \node[above=0pt of s] {$(-\infty, +\infty)$};
  \node[above=0pt of q] {$(-\infty, +\infty)$};
  \node[above=0pt of r, xshift=1cm, yshift=-5pt] {$[20, 100]$};

  %% Transitions
  \path[->]
  (p) edge node       {$+5$} (r)
  (r) edge node       {$-1$} (q)
  (p) edge node[swap] {$+3$} (s)
  (s) edge node[swap] {$-5$} (q)

  (r) edge[loop above] node {$+2$} ()
  ;  
\end{tikzpicture}
  \end{center}
  \caption{A guarded COCA; each state $s$ is labeled with the
    interval $\tau(s)$.}\label{fig:coca:example}
\end{figure}

%We say that $\W$ is \emph{linear} if its structure is a so-called
%linear path scheme, \ie\ there exist nonempty simple paths $\pi_0,
%\ldots, \pi_k$ and simple cycles $\theta_1, \ldots, \theta_k$ such
%that every path of $\W$ is an infix of $\pi_0 \theta_1^* \pi_1
%\theta_2^* \ldots \pi_{k-1} \theta_k^* \pi_k$. Graphically, this
%means that $\W$ has the structure depicted in Figure~\ref{fig:lps}
%(\eg, see~\cite{LS04,BIL09}). 
%It is clear that a linear COCA has a unique input state which is not the output of any transition, and a unique global output state that is not the input of any transition. 
%We write $\In{W}$ and $\Out{W}$ to denote the input and output state, respectively.
%Additionally, we write $\SPaths{\W} = \Paths{\In{\W}}{\Out{\W}}$.

%% It is readily seen that: $k \le |Q|$; every $\pi_i$ and $\pi_j$ have
%% pairwise disjoint states; and similarly every $\delta_i$ and
%% $\delta_j$ have pairwise disjoint states (for all $i \neq j$).

%Intuitively the states in $\beta_i$ form a simple path and $\delta_i$ are cycles on top of the states on this path.

%\subsection{Reachability}

\subsection{Our contribution}

In this work, we study the \emph{reachability problem} that asks the
following question: Given a COCA or a guarded COCA $\W$ with
configurations $p(a)$ and $q(b)$, is there an admissible run from
$p(a)$ to $q(b)$? In other words, by abbreviating ``$\Paths{p}{q}$''
with ``$*$'', the problem asks whether $p(a) \steps{*} q(b)$
holds. For parametric COCAs, the \emph{(existential) reachability
  problem} asks whether $p(a) \steps{*} q(b)$ for some parameter
valuation.

We will establish the following complexity results:

\begin{theorem}\label{main_results}
  The reachability problem is:
  \begin{enumerate}
  \item in $\NC^2$ for COCAs;

%  \item in $\NC^2$ for linear guarded COCAs;

  \item in $\P$ for guarded COCAs; and

  \item $\NP$-hard and in $\SIGTWO$ for parametric COCAs.
  \end{enumerate}
\end{theorem}

Recall that $\SIGTWO$ is the level of the polynomial hierarchy
that corresponds to $\NP^\NP$. Moreover, $\NC$ is the class of
problems solvable in polylogarithmic parallel time, \ie\ $\NC =
\bigcup_{i \geq 0} \NC^i$ where $\NC^i$ is the class of problems
decidable by logspace-uniform families of circuits of polynomial size,
depth $\Oh(\log^i n)$ and bounded fan-in
(\eg, see~\cite{papadimitriou94,ab09} for a more thorough definition). It
is well-known that $\NL \subseteq \NC^2 \subseteq \P$.
%$\NC^1 \subseteq \NL
%\subseteq \NC^2 \subseteq \NC \subseteq \P$.
%Informally, $\NC$ is the
%class of problems which can be efficiently parallelized, in constrast
%to $\P$-complete problems.
%By abuse of language, we 
We also refer to the
functional variant of $\NC^i$ as $\NC^i$.
%, \ie\ when considering
%computational problems rather than decision problems. Recall that
%operations such as $\max$, $\min$, addition, subtraction and
%comparison are computable in $\NC^1$. Moreover, $\NC^i$ is closed
%under the sequential composition of a constant number of $\NC^i$
%procedures, and the combination of a polynomial number of $\NC^i$
%procedures ran in parallel.

%% \begin{theorem}\label{theorem:coca}
%%   The reachability problem for COCA is in $\NC^2$. Moreover,
%%   reachability sets can be computed in $\NC^2$.
%% \end{theorem}

%% \begin{theorem}\label{theorem:guarded}
%%   The reachability problem for guarded COCA is in $\P$. Furthermore,
%%   it belongs to $\NC^2$ for linear guarded COCA.
%% \end{theorem}

%% \begin{theorem}\label{theorem:parametric}
%%   The reachability problem for parametric COCA is $\NP$-hard and
%%   belongs to $\SIGTWO$.
%% \end{theorem}
%% \guillermo{I think in all cases one can construct the reachability relation in
%%   the corresponding complexities}

The two first results of \autoref{main_results} are obtained by
characterising reachability functions and by showing how to
efficiently compute their representation. More precisely, we show:

\begin{proposition}
  Let $\W$ be a COCA or a guarded COCA. It is the case~that:
  \begin{enumerate}
  \item If $\W$ is a COCA, then $\Post{p,q}{a}$ consists of at most
    two intervals whose representations are computable in $\NC^2$;

  \item $\Post{p,q}{a}$ is made of $|\W|^{\bigO(1)}$ intervals, and
    a representation of $\Post{p,q}{a}$ is computable in polynomial time.

%    or more precisely in $\NC^2$ if $\W$ is linear.
  \end{enumerate}
\end{proposition}

%\michael{Please double-check: do you compute a representation of
%  $\mathrm{Post}_{p,q}$ (function) or $\mathrm{Post}_{p,q}(a)$ (set)?}

\noindent To derive the third result, \ie\ the $\SIGTWO$ upper
bound, we borrow ideas from the above technical results to reduce
reachability to determining the truth value of linear arithmetic
$\Sigma_2$-sentences.

%\filip{In Theorems \ref{theorem:guarded} and \ref{theorem:parametric} I wasn't sure if the reachability set can be also computed?}

\section{COCA reachability}
\label{sec:nc2}
In this section, 
%we establish our first result. Namely, 
we prove that
the reachability problem for COCAs belongs in $\NC^2$ by showing how
to compute a representation of $\Post{p,q}{a}$ from some $a$. In the
remainder, we fix a COCA $\V = (Q, T, \tau)$.

\subsection{Testing emptiness}

\newcommand{\first}[1]{\mathrm{first}(#1)}
\newcommand{\last}[1]{\mathrm{last}(#1)}

We first aim to show that deciding whether $\Post{p,q}{a} \neq
\emptyset$ can be checked in $\NC^2$. To this end, we first state some
simple graph properties checkable in $\NC^2$. For a path $\pi$, let us
write $\first{\pi}$ (resp.\ $\last{\pi}$) to denote the first
(resp.\ last) index such that $\Effect{\pi_i} \neq 0$ if any, and
$\first{\pi} = \last{\pi} \defeq \infty$ if none. We naturally extend
the notations $\Effectp{\pi}$, $\Effectn{\pi}$, $\Effect{\pi}$,
$\first{\pi}$ and $\last{\pi}$ to weighted multigraphs.

This lemma follows from standard results on $\NC^2$:

\begin{restatable}{lemma}{lemmaGraphReach}\label{lemma:graph-reach-nc2}
  Let $G = (Q, E)$ be a weighted multigraph whose weights are encoded
  in binary, and let $p, q \in Q$ be nodes.  
  Deciding
  whether $S = \emptyset$ is in $\NC^2$,
  where
  $S$ is the set of paths $\pi \in \Paths{p}{q}$ that satisfy a fixed
  subset of these conditions\footnote{The set of conditions may be empty, in which case $S = \Paths{p}{q}$.}:
  \begin{enumerate}[label=(\alph*)]
  \item $\Effectp{\pi} \neq 0$ (resp.\ $\Effectn{\pi} \neq
    0$);\label{itm:has-pos}
    
  % \item $\pi$ contains a positive (resp.\ negative) edge;\label{itm:has-pos}

  \item $\Effectp{\pi} = 0$ (resp.\ $\Effectn{\pi} = 0$);\label{itm:no-neg}
      
  % \item $\pi$ contains no edge with negative (resp.\ positive)
  % weight;\label{itm:no-neg}
  
  \item $\Effect{\first{\pi}} < 0$ (resp.\ $\Effect{\first{\pi} > 0}$);\label{itm:first:pos}
    
  \item $\Effect{\last{\pi}} < 0$ (resp.\ $\Effect{\last{\pi} > 0}$).\label{itm:last:neg}

  \end{enumerate}
  Furthermore, for any such set $S$, the following value can be
  computed in $\NC^2$: $\mathrm{opt}\{w(\pi) \mid \pi \in S \text{ and
  } \abs{\pi} \leq \abs{Q})\}$, where $\mathrm{opt} \in \{\min,
  \max\}$ and $w \in \{\Delta^+, \Delta^-\}$.
\end{restatable}

In essence, conditions~\ref{itm:has-pos}--\ref{itm:last:neg} above can be
checked through preprocessing and standard graph reachability,
\eg\ for \ref{itm:first:pos}, we make a copy of $G$ restricted to zero
transitions, which branches into a full copy of $G$ via negative
(resp.\ positive) transitions. Optimising the value of such paths can
be done by a typical ``divide-and-conquer procedure'' that optimizes
paths of length $1, 2, 4, \ldots, |Q|$ (hence $\log |Q|$ levels of
$\NC^1$ computations).

%\philip{the following proposition/its proof can probably be shortened by a large factor}
%\michael{Done}

\begin{lemma}\label{lemma:adm-paths}
  Let $a \in \Q$, $p, q \in Q$, and $\pi \in \Paths{p}{q}$. 
  We have $a \in \enab{\pi}$ iff $a \in \tau$
  and any of these conditions hold:
  \begin{enumerate}[label=(\alph*)]
  \item $a \notin \{\inf \tau, \sup \tau\}$;\label{itm:adm:a}

  \item $a = \inf \tau = \sup \tau$, $\first{\pi} = \infty$;\label{itm:adm:b}

  \item $a = \inf \tau < \sup \tau$, $\first{\pi} \neq \infty \implies
    \effect{\pi_{\first{\pi}}} > 0$;\label{itm:adm:c}

  \item $a = \sup \tau > \inf \tau$, $\first{\pi} \neq \infty \implies
    \effect{\pi_{\first{\pi}}} < 0$.\label{itm:adm:d}
    
  %% \item $\inf \tau < a < \sup \tau$;\label{itm:adm-notsupinf}
    
  %% \item $\Effectn{\pi} = \Effectp{\pi} = 0$;\label{itm:adm-zero}
    
  %% \item $a \neq \sup \tau$ and $i_{+} < i_{-}$;\label{itm:adm-notsup}
    
  %% \item $a \neq \inf \tau$ and $i_{-} < i_{+}$.\label{itm:adm-notinf}
  \end{enumerate}
\end{lemma}

%\philip{\ref{itm:adm-notsupinf} is unnecessary, could be corollary. though it helps with the proof a bit.}

\begin{proof}
  Having $a \in \tau$ is obviously necessary, so we assume it
  holds throughout the proof.

  $\Leftarrow$) We proceed by induction on $|\pi|$. If $|\pi| = 0$, then
  the claim is trivial as the empty path is admissible
  from $a$.
  Assume $|\pi| = n > 0$
  and $\pi$ satisfies a condition. Let $t \defeq \pi_1$
  and $\sigma \defeq \FromTo{\pi}{2}{n}$. If~\ref{itm:adm:a} holds,
  then $a \steps{\beta t} a'$ for some $a' \in \tau \setminus \{\inf
  \tau, \sup \tau\}$ and sufficiently small $\beta \in (0,
  1]$. If~\ref{itm:adm:b} holds, then $a \steps{t} a' = a$ as
    $\Effect{t} = 0$. If~\ref{itm:adm:c} or~\ref{itm:adm:d} holds,
    then either $a \steps{t} a' = a$ if $\Effect{t} = 0$, or $a
    \steps{\beta t} a'$ for some $a' \in \tau \setminus \{\inf \tau,
    \sup \tau\}$ and sufficiently small $\beta \in (0, 1]$
      otherwise. In all cases, $\sigma$ satisfies one of the
      conditions w.r.t.\ value $a'$. Thus, were are done by the induction
      hypothesis.

  $\Rightarrow$) Towards a contradiction, let us assume that $a \in
      \enab{\pi}$ and that no condition is satisfied. If $a = \inf
      \tau = \sup \tau$ and $\first{\pi} \neq \infty$, then there is obviously a
      contradiction. Otherwise, either (i)~$a = \inf \tau$ and
      $\Effect{\pi_{\first{\pi}}} < 0$; or (ii)~$a = \sup \tau$ and
      $\Effect{\pi_{\first{\pi}}} > 0$. We only consider~(ii) as~(i) is
      symmetric. Let $a \steps{\FromTo{\pi}{1}{\first{\pi}-1}} a'$. We have $a'
      = a$ by definition of $\first{\cdot}$. Moreover, $a' + \beta \cdot
      \Effect{\pi_{\first{\pi}}} > a' = a = \sup \tau$ for any $\beta \in (0,
      1]$. Since exceeding $\sup \tau$ is forbidden, we obtain the
        contradiction $a \notin \enab{\FromTo{\pi}{1}{\first{\pi}}} \supseteq
        \enab{\pi}$.
\end{proof}

\begin{corollary}\label{lemma:reachability-list}
  Given $a \in \Z$ and $p, q \in Q$, deciding whether $a \in
  \enab{\Paths{p}{q}}$, or equivalently $\Post{p,q}{a} \neq
  \emptyset$, is in~$\NC^2$.
\end{corollary}

%\guillermo{Is there a reason $a$ is an integer here?}
%\michael{Yes, so that the input can be fed to a circuit and use standard machinery over integers without caring about the representation of rationals.}

\begin{proof}
  We report ``empty'' if $a \notin \tau$. Otherwise, let
  \begin{align*}
    S_0 &\defeq \{\pi \in \Paths{p}{q} \mid \Effectp{\pi} =
    \Effectn{\pi} = 0\}, \\
    S_+ &\defeq \{\pi \in \Paths{p}{q} \mid \first{\pi} \neq \infty \implies
    \Effect{\pi_{\first{\pi}}} > 0\}, \\
    S_- &\defeq \{\pi \in \Paths{p}{q} \mid \first{\pi} \neq \infty \implies
    \Effect{\pi_{\first{\pi}}} < 0\}.
  \end{align*}

  By \autoref{lemma:adm-paths}, it suffices if one the following
  holds:
  \begin{enumerate}[label=(\alph*)]
  \item $a \notin \{\inf \tau, \sup \tau\}$ and $\Paths{p}{q} \neq
    \emptyset$;

  \item $a = \inf \tau = \sup \tau$ and $S_0 \neq \emptyset$;

  \item $a = \inf \tau < \sup \tau$ and $S_+ \neq \emptyset$;

  \item $a = \sup \tau > \sup \tau$ and $S_- \neq \emptyset$.
  \end{enumerate}
  All of the above can be checked in $\NC^2$ by
  \autoref{lemma:graph-reach-nc2}.
\end{proof}

\subsection{Characterisation of reachability sets}

As a step towards computing a representation of 
$\Post{p,q}{a}$, we characterise $\Post{p,q}{a}$ in terms of its
closure. To this end, we note that admissible runs remain
admissible whenever they are scaled down. Consequently,  
 $\cl{\Post{p,q}{a}}$ is a closed interval that differs from
$\Post{p,q}{a}$ in at most three points.

\begin{restatable}[{Adapted from~\cite[Lemma 4.2(c)]{bh17}}]{proposition}{propScaleDown}\label{claim:scaling-down-gives-run}
  Let $\beta \in (0, 1]$ and let $\rho$ be an admissible run from
    configuration $p(a)$. It is the case that run $\beta \rho$ is also
    admissible from $p(a)$.
\end{restatable}

% \begin{proof}
% Let $\rho = \alpha_1 t_1 \cdots \alpha_n t_n$ be an admissible run from $p(a)$ to $q(b)$. Then $\beta \rho$ is a run from $p(a)$ to $q(c)$, where $c = a + (b-a)\beta$. Notice that since $c$ is between $a\ge 0$ and $b\ge 0$ then $c \ge 0$. In particular it means that $\rho_{1..j}$ is a run from $p(a)$ that ends in a nonnegative value for all $j = 1,\ldots,n$. Thus $\beta \rho$ is an admissible run from $p(a)$.
% \end{proof}

\begin{restatable}{lemma}{lemmaIntervals}\label{lemma:intervals}
  For every $b \in \cl{\Post{p,q}{a}}$, it is the case that $(a, b)
  \subseteq \Post{p,q}{a}$ and $(b, a) \subseteq \Post{p,q}{a}$.
\end{restatable}

The above lemma holds as a run from $a$ to $b$ can be scaled
to reach an arbitrary value from $(a, b)$ and remain admissible.

\begin{corollary}\label{corollary:intervals}
  Set $\cl{\Post{p,q}{a}}$ is a closed interval. Moreover,
  $\cl{\Post{p,q}{a}} \setminus \Post{p,q}{a} \subseteq \{\inf
  \cl{\Post{p,q}{a}}, a, \sup \cl{\Post{p,q}{a}}\}$.
\end{corollary}

\begin{proof}
  Let $b \defeq \inf \cl{\Post{p,q}{a}}$ and $c \defeq \sup
  \cl{\Post{p,q}{a}}$. For the sake of contradiction, suppose there
  is some $v \in \cl{\Post{p,q}{a}} \setminus \Post{p,q}{a}$ such
  that $v \notin \{b, a, c\}$. By \autoref{lemma:intervals}, we have
  $(a, b) \cup (a, c) \cup (b, a) \cup (c, a) \subseteq \Post{p,q}{a}
  \subseteq \cl{\Post{p,q}{a}}$.  Since $v \in (b, c) \setminus
  \{a\}$, we obtain $v \in \Post{p,q}{a}$, which is a contradiction.
\end{proof}

%% By \autoref{corollary:intervals}, $\Post{p,q}{a}$ is a union of at
%% most two intervals with endpoints in $S \defeq \{\inf
%% \cl{\Post{p,q}{a}}, a, \sup \cl{\Post{p,q}{a}}\}$. Thus, a
%% representation of $\Post{p,q}{a}$ can be computed in two steps. First,
%% we identify the endpoints of the interval $\cl{\Post{p,q}{a}}$, and
%% then we determine the elements of $S$ that belong to $\Post{p,q}{a}$.

\subsection{Identifying the endpoints}

We now show that a representation of the interval $\cl{\Post{p,q}{a}}$
can be obtained by identifying its endpoints in $\NC^2$.
Some simple observations 
follow from \autoref{claim:scaling-down-gives-run} and 
\autoref{lemma:intervals}:

\begin{restatable}{proposition}{propObs}\label{prop:obs}
  The following statements hold:
  \begin{enumerate}[label=(\alph*)]
  \item If $\Post{p,q}{\inf \tau} \neq \emptyset$, then
    $\inf{\cl{\Post{p,q}{\inf \tau}}} = \inf \tau$.\label{itm:a=inf}

  \item If $\Post{p,q}{\sup \tau} \neq \emptyset$, then
  $\sup{\cl{\Post{p,q}{\sup \tau}}} = \sup \tau$.\label{itm:a=sup}

  \item Let $v \in \tau \setminus \{\inf \tau, \sup \tau\}$ and let
    $\rho$ be a run. There exists $\epsilon \in (0, 1]$ such that for
    all $\beta \in (0, \epsilon]$ there exists $v_\beta > 0$ such that
      $v \steps{\beta \rho} v_\beta$. Moreover, $\lim_{\beta \to 0}
      v_{\beta} = v$.\label{claim:scale}
  \end{enumerate}
\end{restatable}

%\guillermo{In the proof we use a but I think we mean inf of $\tau$, right?}
%\michael{Yes, should be fixed now.}

The two forthcoming lemmas characterise the endpoints of
$\cl{\Post{p,q}{a}}$ through so-called admissible cycles. We say that
a cycle $\theta$ is $(a, p, q)$-\emph{admissible} if its first
transition $t$ satisfies $\Effect{t} \neq 0$, $a \in
\enab{\Paths{p}{\In{t}}}$ and $\Paths{\In{t}}{q} \neq \emptyset$. We
say that such an admissible cycle is \emph{positive} if
$\Effect{t} > 0$, and \emph{negative} if $\Effect{t} < 0$. Such cycles
can be iterated to approach the endpoints of $\tau$, by scaling all
transitions but $t$ arbitrarily close to zero.

%% \philip{it might be helpful to state
%% that this is equivalent to having a cycle where
%% any transition is positive(/negative) 
%% on any admissible path from $p$ to $q$.
%% This is what allows one to assume cycle-freeness of runs with maximal(/minimal) effect}

\begin{restatable}{lemma}{lemmaSupfInf}\label{lemma:supfinf}
  If $\Post{p,q}{a} \neq \emptyset$ and $\V$ has an $(a, p,
  q)$-admissible cycle $\theta$, then the following holds:
  \begin{enumerate}[label=(\alph*)]
  \item\label{l:e} $\inf \cl{\Post{p,q}{a}} = \inf \tau$, if $\theta$
    is negative;

  \item\label{l:f} $\sup \cl{\Post{p,q}{a}} = \sup \tau$, if $\theta$
    is positive.
  \end{enumerate}
\end{restatable}

\begin{lemma}\label{lemma:effinite}
  Let $\Post{p,q}{a} \neq \emptyset$, $b \defeq \inf
  \cl{\Post{p,q}{a}}$ and $c \defeq \sup \cl{\Post{p,q}{a}}$. If $\V$
  has no $(a, p, q)$-admissible cycle which~is:
  \begin{enumerate}[label=(\alph*)]
  \item\label{l:negative} negative,
    then $b \neq -\infty$ and $b = \max(\inf \tau, a +
    \min\{\Effectn{\pi} \mid \pi \in \Paths{p}{q}, a \in
    \enab{\pi}\})$;

  \item\label{l:positive} positive,
    then $c \neq +\infty$ and $c = \min(\sup \tau, a +
    \max\{\Effectp{\pi} \mid \pi \in \Paths{p}{q}, a \in
    \enab{\pi}\})$.
  \end{enumerate}
\end{lemma}

\begin{proof}
  We only prove~\ref{l:positive} as~\ref{l:negative} is
  symmetric. Assume $\V$ has no positive $(a, p, q)$-admissible
  cycle. Let $D^+ \defeq \{\Effectp{\pi} \mid \pi \in \Paths{p}{q}, a
  \in \enab{\pi} \}$. We show that $\max D^+$ is well-defined. For the
  sake of contradiction, suppose that $D^+$ is infinite. By a
  pigeonhole argument, we obtain a run $\rho$ admissible from $a$ and
  such that $\rho$ contains at least two occurrences of a transition
  $t$ with $\Effect{t} > 0$. Let $\upath{\rho} = \pi t \pi' t \pi''$
  where $\pi, \pi', \pi''$ are paths. The cycle $\theta \defeq t \pi'$
  is a positive admissible cycle, which yields a contradiction.

  Note that $c \leq \min(\sup \tau, a + \max D^+)$, so $c \neq
  +\infty$. It remains to show that $c = \min(\sup \tau, a + \max
  D^+)$. Let $\pi \in \Paths{p}{q}$ be such that $a \in \enab{\pi}$
  and $\Effectp{\pi} = \max D^+$. By definition, there exists a run
  $\rho = \alpha_1 t_1 \cdots \alpha_n t_n$ admissible from $p(a)$ and
  such that $\upath{\rho} = \pi$. Since $a \in \tau$, there exists
  $\lambda \in (0, 1]$ such that $a + \lambda \cdot \max D^+ =
  \min(\sup \tau, a + \max D^+)$. For all $\epsilon \in (0, 1)$, let
  $\rho_\epsilon \defeq \alpha_1' t_1 \cdots \alpha_n' t_n$ be the run
  such that
  \[
  \alpha_i' \defeq
  \begin{cases}
    (1 - \epsilon) \cdot \lambda &
    \text{if } \Effect{t_i} \geq 0, \\
    \epsilon \cdot (1 / |\Effect{t_i}|) \cdot (1 / n) &
    \text{otherwise}.
  \end{cases}
  \]
  Informally, if were allowed to scale transitions by $0$, then we
  would be done by using $\rho_0$ from $a$, as it would never decrease
  and reach exactly $a + \lambda \cdot \max D^+ = \min(\sup \tau, a +
  \max D^+)$.

  Formally, we choose a small $\epsilon \in (0, 1]$ as
  follows. If $a > \inf \tau$, then we pick $\epsilon$ so that $a -
  \epsilon \geq \inf \tau$. Otherwise, we pick $\epsilon$ so that $(1
  - \epsilon) \cdot\lambda \geq \epsilon$. We claim that the run
  $\rho_\delta$ is admissible from $a$ for every $\delta \in (0,
  \epsilon]$.
  First note that the top guard is never exceeded since $a +
  \Effectp{\rho_\delta} = a + (1 - \delta) \cdot \lambda \cdot \max
  D^+ \leq a + \lambda \cdot \max D^+ \leq \sup \tau$.  Let us now
  consider the bottom guard.

  If $a > \inf \tau$, then $a + \Effectn{\rho_\delta} \geq a - \delta
  \geq a - \epsilon \geq \inf \tau$. Otherwise, if $a = \inf \tau$,
  then either $\Effectn{\rho_\delta} = \Effectp{\rho_\delta} = 0$, in
  which case admissibility is trivial, or the first transition $t_i$
  such that $\Effect{t_i} \neq 0$ is such that $\Effect{t_i} \geq
  1$. In that case, the following holds for every $j \geq i$:
  \begin{align*}
    a + \Effect{\FromTo{\rho_\delta}{1}{j}}
    &\geq a + (1 - \delta) \cdot \lambda \cdot \Effect{t_i}
    + \Effectn{\FromTo{\rho_\delta}{i+1}{j}} \\
    &\geq a + (1 - \epsilon) \cdot \lambda
    + \Effectn{\FromTo{\rho_\delta}{i+1}{j}} \\ 
    &\geq a + (1 - \epsilon) \cdot \lambda
    - \epsilon \geq a = \inf \tau.
  \end{align*}
  This shows the admissibility of $\rho_\delta$. Thus, for all $\delta
  \in (0, \epsilon]$, we have $a \steps{\rho_\delta} a_\delta$ where
    $a_\delta \geq \min(\sup \tau, a + \max D^+) - \delta \cdot
    \lambda \cdot \max D^+$. We are done since $\lim_{\delta \to 0}
    a_\delta = \min(\sup \tau, a + \max D^+)$.
\end{proof}

In the forthcoming propositions, we show how the previous
characterisations can be turned into $\NC^2$ procedures.

\begin{lemma}\label{lemma:adm-cycle-ef}
  On input $a \in \Z$ and $p, q \in Q$, deciding if $\V$ has a
  positive or negative $(a, p, q)$-admissible cycle is in $\NC^2$.
\end{lemma}

\begin{proof}
  We consider the positive case; the negative one is
  symmetric. Testing whether there is a positive $(a, p,
  q)$-admissible cycle beginning with a transition $t$ with
  $\effect{t} > 0$ amounts to testing whether (i)~$a \in
  \enab{\Paths{p}{\In{t}}}$, (ii)~$\Paths{\In{t}}{q} \neq \emptyset$,
  and (iii)~$\Paths{\Out{t}}{\In{t}} \neq \emptyset$. Condition~(i)
  can be checked in $\NC^2$ by
  \autoref{lemma:reachability-list}. Conditions~(ii) and~(iii) are
  graph reachability queries which can be tested in $\NL \subseteq
  \NC^2$. There are at most $\abs{T}$ transitions with positive
  effect, so the conditions can be tested in parallel for each $t \in T$.
\end{proof}

%  \newcommand{\SimPaths}[2]{\text{Simple}\Paths{#1}{#2}}

% For states $p, q \in Q$, let us write $\SimPaths{p}{q}$ to denote
% the set $\SimPaths{p}{q} \defeq \{\pi \in \Paths{p}{q} \mid \text{$\pi$ is a simple path}\}$.

% \begin{corollary}
%   Let $a \in \tau, p, q \in Q$. 
%   Each of the following values can be computed in $\NC^2$:
%   \begin{enumerate}[label=(\alph{*})]
%     \item $\Effect{\pi}$ for some $\pi \in \Paths{p, q}$ such that $\Effect{\pi} \geq max \{\Effectp{\pi'} \mid \pi \in \SimPaths{p}{q} \wedge a \in \enab{\pi'}\}$\label{itm:maxp}
%     \item $\Effect{\pi}$ for some $\pi \in \Paths{p, q}$ $\Effect{\pi} \geq max \{\Effectp{\pi'} \mid \pi' \in \SimPaths{p}{q} \wedge a \in \enab{\pi'} \wedge \Effectn{\pi'} = 0\}$\label{itm:maxp_n0}
%     \item $y \in \Q$ with $y \leq min \{\Effectn{\pi} \mid \pi \in \SimPaths{p}{q} \wedge a \in \enab{\pi}\}$\label{itm:maxn}
%     \item $y' \in \Q$ with $y' \leq min \{\Effectn{\pi} \mid \pi \in \SimPaths{p}{q} \wedge a \in \enab{\pi} \wedge \Effectp{\pi} = 0\}$\label{itm:maxn_p0}
%   \end{enumerate}
% \end{corollary}
% \begin{proof}
%   We argue the case for \ref{itm:maxp} and \ref{itm:maxp_n0}.
%   The case for \ref{itm:maxn} and \ref{itm:maxn_p0} is symmetric.

%   \ref{itm:maxp} follows from 
% \end{proof}

\begin{proposition}\label{computing:ef}
  On input $a \in \Z$ and states $p, q$, the values $\inf
  \cl{\Post{p,q}{a}}$ and $\sup \cl{\Post{p,q}{a}}$ can be computed in
  $\NC^2$.
\end{proposition}

\begin{proof}
  We explain how to compute
  $c \defeq \sup \cl{\Post{p,q}{a}}$ in $\NC^2$. 
   The procedure for $\inf \cl{\Post{p,q}{a}}$ is
  symmetric. By \autoref{lemma:reachability-list}, testing whether
  $\Post{p,q}{a} = \emptyset$ is in $\NC^2$. If it holds, then trivially
  $c = -\infty$. Otherwise, assume that $\Post{p,q}{a} \neq
  \emptyset$, and hence $a \in \tau$. Additionally, if $a = \sup \tau$,
  then $c = \sup \tau$ by \autoref{prop:obs}\ref{itm:a=sup}. So we assume $a < \sup \tau$.
  
  By \autoref{lemma:adm-cycle-ef}, 
  it can be decided in $\NC^2$ whether there exists a positive
  $(a, p, q)$-admissible cycle. 
  If such a cycle exists, then $c = \sup \tau$ by
  \autoref{lemma:supfinf}. Otherwise, by \autoref{lemma:effinite}, we
  have $c = \min(\sup \tau, a + \max D^+)$ where $D^+ \defeq
  \{\Effectp{\pi} \mid \pi \in \Paths{p}{q}, a \in \enab{\pi}\}$.

  Consider a path $\pi \in \Paths{p}{q}$ that satisfies $a \in
  \enab{\pi}$ and which can be decomposed as $\pi = \sigma \theta
  \sigma'$ where $\theta$ is a cycle. As $\theta$ is $(a, p, q)$-admissible
  by definition, it cannot contain a positive transition.
  Otherwise, $\mathcal{V}$ would admit a positive $(a, p,
  q)$-admissible cycle, which is a contradiction.
  Hence $\Effectp{\pi} =
  \Effectp{\sigma \sigma'}$.
   We show that $a
  \in \enab{\sigma \sigma'}$. Recall that $a \neq \sup \tau$. 
  Hence $a \in
  \enab{\pi}$ follows from 
  \autoref{lemma:adm-paths}\ref{itm:adm:a} or \ref{itm:adm:c}.
  Note that neither condition can be violated by removing a
  nonpositive transition from $\pi$,
  so $a \in \enab{\sigma \sigma'}$.

  The above shows that there is a simple path $\pi_{max}$ such that
  $\max D^+ = \Effectp{\pi_{max}}$ and $a \in \enab{\pi_{max}}$.
  Therefore, to obtain $\max D^+$, it is sufficient to compute $\max E^+$
  where $E^+ \defeq \{\Effectp{\pi} \mid \pi \in \Paths{p}{q}, a \in
  \enab{\pi}, \abs{\pi} \leq \abs{Q}\}$.
  
  Now, let us make a case distinction based on whether $a = \inf \tau$.
  Assume this is true. By \autoref{lemma:adm-paths}, $\max E^+$ equals $\max($
  \begin{alignat*}{2}
    &\max\{\Effectp{\pi} \mid \pi \in \Paths{p}{q},
    \first{\pi} = \infty,
    \abs{\pi} \leq \abs{Q}\}, \\   
    &\max\{\Effectp{\pi} \mid \pi \in \Paths{p}{q},
    \Effect{\pi_{\first{\pi}}} > 0, \abs{\pi} \leq \abs{Q}\}\,
  \end{alignat*}$)$.
  By \autoref{lemma:graph-reach-nc2}, the above is the
  maximum of two values that can be computed in $\NC^2$. If $a \neq
  \inf \tau$, then we have $a \in \tau \setminus \{\inf \tau, \sup
  \tau\}$, so by \autoref{lemma:adm-paths}\ref{itm:adm:a} all paths
  are admissible from $a$, and hence $\max E^+ = \max\{\Effectp{\pi} \mid \pi \in
  \Paths{p}{q} \wedge \abs{\pi} \leq \abs{Q}\}$. This value can be
  computed in $\NC^2$ by \autoref{lemma:graph-reach-nc2}.
\end{proof}

\subsection{Computing the representation}

To obtain a representation of $\Post{p,q}{a}$, it remains to
explain how to check in $\NC^2$ which of the three limit elements $\inf
\cl{\Post{p,q}{a}}$, $\sup \cl{\Post{p,q}{a}}$ and $a$ belong to
$\Post{p,q}{a}$.

\begin{proposition}\label{prop:a_nc2}
  Testing whether $a \in \Post{p,q}{a}$ is in $\NC^2$.
\end{proposition}

\begin{proof}
  By \autoref{lemma:reachability-list}, $\Post{p,q}{a} \neq \emptyset$
  can be tested in $\NC^2$. Thus, we assume that it is nonempty. It
  is easy to show that $a \in \Post{p,q}{a}$ iff at least one of these
  conditions~holds: \begin{enumerate}[label=(\alph*)] \item\label{a:1}
  there exists a path $\pi \in \Paths{p}{q}$ whose transitions are all
  zero, \ie\ $\Effect{\pi} = \Effectp{\pi} = \Effectn{\pi} = 0$;
    
  \item\label{a:2} there exist $\pi \in \Paths{p}{q}$ and $i,
    j$ such that $\Effect{\pi_i} > 0$ and $\Effect{\pi_j} < 0$. If $a
    = \inf \tau$, then we also require $\Effect{\pi_k} = 0$ for
    all $k < i$ and $k > j$. Similarly, if $a = \sup \tau$, then we
    also require $\Effect{\pi_k} = 0$ for all $k < j$ and $k > i$.
  \end{enumerate}
  %
  % Note that~\ref{a:2} is trivially unsatisfiable if $a = \inf \tau =
  % \sup \tau$.
    
  It remains to argue that both conditions can be checked in $\NC^2$.
  We can check condition~\ref{a:1} in $\NC^2$ via
  \autoref{lemma:graph-reach-nc2}\ref{itm:no-neg}. If $a \not\in
  \{\inf \tau, \sup \tau\}$, then condition~\ref{a:2} can be checked
  in $\NC^2$ via \autoref{lemma:graph-reach-nc2}\ref{itm:has-pos}. If
  $a \in \{\inf \tau, \sup \tau\}$, then condition~\ref{a:2} can be
  checked in $\NC^2$ via
  \autoref{lemma:graph-reach-nc2}\ref{itm:last:neg} and
  \autoref{lemma:graph-reach-nc2}\ref{itm:first:pos}.
\end{proof}

\begin{proposition}\label{prop:endpoints_nc2}
  On input $a \in \Z$ and $p, q \in Q$, computing
  $\Post{p,q}{a} \cap \{\inf \cl{\Post{p,q}{a}}, \sup \cl{\Post{p,q}{a}}\}$
  is in $\NC^2$.
\end{proposition}

\begin{proof}
  By \autoref{lemma:reachability-list}, we can check if $\Post{p,q}{a}
  \neq \emptyset$ in $\NC^2$. Thus, we assume it is
  nonempty. Let $b \defeq \inf \cl{\Post{p,q}{a}}$ and $c \defeq \sup
  \cl{\Post{p,q}{a}}$, which can be computed in $\NC^2$ by
  \autoref{computing:ef}. We check whether $b = c$. If it is,
  we return $\{b, c\}$ since $\Post{p,q}{a} \neq
  \emptyset$. Otherwise, we explain how to check whether $c \in
  \Post{p,q}{a}$; the case of $b$ can be handled symmetrically. We
  assume that $b < a < c$, as \autoref{prop:a_nc2} handles the
  case $a \in \{b, c\}$ when checking membership of $a$.

  If $c \not\in \tau$, then $c \not\in \Post{p,q}{a}$. Otherwise, it
  can be shown that $c \in \Post{p,q}{a}$ iff there is a state $r \in Q$ and
  a path $\sigma \in \Paths{r}{q}$ that satisfy $\Effectp{\sigma} >
  0$, $\Effectn{\sigma} = 0$ and either of the following holds:
  \begin{enumerate}[label=(\roman*)]
  \item there exists a path $\sigma' \in \Paths{p}{r}$ such that
    $|\sigma|, |\sigma'| \leq |Q|$, $\Effectn{\sigma'} = 0$ and
    $\Effectp{\pi} \geq c - a$ where $\pi \defeq \sigma'
    \sigma$;\label{itm:sup:simple:a}

  \item there exists a path $\sigma' \in \Paths{p}{r}$ such that
    $|\sigma|, |\sigma'| \leq |Q|$ and $\Effectp{\pi} > c - a$ where
    $\pi \defeq \sigma' \sigma$;\label{itm:sup:simple:b}

  \item there is a positive $(a, p, r)$-admissible cycle
    $\theta$.\label{itm:sup:cycle}
  \end{enumerate}

  It remains to show that the conditions can be tested in~$\NC^2$. There are $\abs{Q}$ choices for state $r$, so we can test
  the conditions for all choices in parallel. Let $S \defeq \{\sigma
  \in \Paths{r}{q} \mid \Effectp{\sigma} > 0 \text{ and }
  \Effectn{\sigma} = 0\}$. We first check whether $S \neq
  \emptyset$, which can be done in $\NC^2$ by
  \autoref{lemma:graph-reach-nc2}\ref{itm:has-pos}--\ref{itm:no-neg}.
  Moreover, condition~\ref{itm:sup:cycle} can be checked in $\NC^2$ by
  \autoref{lemma:adm-cycle-ef}. We proceed as follows to
  check~\ref{itm:sup:simple:a}. Let
  \begin{alignat*}{2}
    W
    &\defeq \{\Effectp{\sigma} &&\mid \sigma \in S,
    |\sigma| \leq |Q|\}, \\
    W'
    &\defeq \{\Effectp{\sigma'} && \mid \sigma' \in \Paths{p}{r},
    |\sigma'| \leq |Q|,\Effectn{\sigma'} = 0\}.
  \end{alignat*}
  By \autoref{lemma:graph-reach-nc2}, we can compute $m \defeq \max W
  + \max W'$ in $\NC^2$ and check that $m \geq c - a$.
  Lastly, we define $W'' \defeq
  \{\Effectp{\sigma'} \mid \sigma' \in \Paths{p}{r}, |\sigma'| \leq
  |Q|\}$ and test whether $\max W + \max W'' > c - a$ to verify
   condition~\ref{itm:sup:simple:b}.
\end{proof}

\begin{theorem}\label{thm:nc2}
  Given $a, a' \in \Z$ and $p, q \in Q$, the following can be done in
  $\NC^2$: obtaining a representation of $\Post{p,q}{a}$ and testing
  whether $a' \in \Post{p,q}{a}$.
\end{theorem}

\begin{proof}
  By \autoref{computing:ef}, we can compute $b \defeq \inf
  \cl{\Post{p,q}{a}}$ and $c \defeq \sup \cl{\Post{p,q}{a}}$ in
  $\NC^2$. By \autoref{corollary:intervals}, \autoref{prop:a_nc2} and
  \autoref{prop:endpoints_nc2}, the set $S \defeq \cl{\Post{p,q}{a}}
  \setminus \Post{p,q}{a}$, of size at most three, can be computed in
  $\NC^2$. By \autoref{corollary:intervals}, this yields the
  representation $\Post{p,q}{a} = [b, c] \setminus S$. Thus,
  $a' \in \Post{p,q}{a}$ iff $b \leq a' \leq c$ and $a' \not\in S$.
\end{proof}

\subsection{Equality tests}

\newcommand{\eqguard}{\phi}

A \emph{COCA with equality tests} is a tuple $\V = (Q,
T, \allowbreak \tau, \eqguard)$, where $(Q,T,\tau)$ is a COCA
and $\eqguard \colon Q \to \{[z, z] \mid z \in \Z\} \cup \Q$. We say
that a run of $\V$ is \emph{admissible} if each of its configurations
$q(a)$ satisfies $a \in \tau \cap \eqguard(q)$.

Using the previous results, we can extend the $\NC^2$ membership of
the reachability problem $p(a) \steps{*} q(b)$ to COCA
with equality tests. The proof relies on
the fact that each equality test is passed by exactly one
configuration. For this reason, we can construct a reachability graph
between equality tests using a quadratic number of COCA
reachability queries.

Let us assume that $p$
has no incoming edges; if it does,
we can simply add a new initial state $p'$ and
a single transition $(p', 0, p)$.
Similarly, we can assume $q$ has no outgoing edges.

We will reason about reachability in $\V$ where we avoid all equality
tests.  For every states $p', q' \in Q$, let us define the
COCA $\V_{p',q'} \defeq (Q_{p',q'}, T_{p',q'})$ where
$Q_{p',q'} \defeq \{s \in Q \mid \eqguard(s) = \Q\} \cup \{p', q'\}$.
We treat $p'$ as a dedicated input state, and $q'$ as a dedicated
output state. That is, \[T_{p',q'} \defeq \{t \in T \mid \In{t} \in
Q_{p',q'}\setminus \{q'\}, \Out{t} \in Q_{p',q'} \setminus \{p'\}\}.\]

Let us define a graph $\mathcal{G} \defeq (V, E)$, where
$V \defeq \{p(a), q(b)\} \allowbreak \cup \allowbreak  \{r(z) \mid r \in Q, \phi(r) = [z, z],
z \in \tau\}$. If $a \notin \tau \cap \phi(p)$ or
$b \notin \tau \cap \phi(q)$, then we trivially conclude that $p(a)$
cannot reach $q(b)$. Hence, $\abs{V} \leq \abs{Q}$
holds. Intuitively, the nodes of $\mathcal{G}$ correspond to the
initial and final configurations, plus, for each equality test, the
configuration that passes this test. Let us define $E \defeq \{(p'(x),
q'(y)) \mid p'(x) \steps{*} q'(y) \text{ in $\V_{p', q'}$}\}$.

\begin{restatable}{lemma}{lemmaReachEquivalent}\label{lemma:reach-equivalent}
  It is the case that $p(a) \steps{*} q(b)$ in $\V$ if and only if
  there is a path from $p(a)$ to $q(b)$ in $\mathcal{G}$.
\end{restatable}

An edge between two nodes $p'(a')$ and $q'(b')$ in $\mathcal{G}$
implies reachability between $p'(a')$ and $q'(b')$ in
$\V_{p',q'}$, from which reachability between 
$p'(a')$ and $q'(b')$ in $\V$ follows.
Hence the ``if'' direction follows by the fact that reachability is transitive.

Further, a run from $p(a)$ to $q(b)$ in $\V$ can
be split into runs that each
witness reachability between two equality tests
(or the initial/target configuration),
but avoid all equality tests in intermediate configurations.
Therefore, these runs are also runs in $\V_{p',q'}$
for some $p',q'$, so the corresponding edges exist in $\mathcal{G}$.
Hence the ``only if'' direction follows.

%% \begin{lemma}\label{lemma:construction-nc2}
%%     The graph $\mathcal{G}$ can be constructed
%%     in $\NC^2$.
%% \end{lemma}

%% \begin{proof}
%%     There are at most $\abs{Q}$ nodes in $\mathcal{G}$,
%%     and hence at most $\abs{Q}^2$ edges.
%%     Note that an edge $(p'(x), q'(y))$ is present
%%     if and only if $p'(x) \steps{*} q'(y)$ in $\V_{p', q'}$.
%%     Further, $\V_{p',q'}$ is an unguarded COCA,
%%     so per \autoref{???},
%%     \philip{todo: reference nc2 lemma for unguarded}
%%     inclusion of each edge can be decided in $\NC^2$.
%%     By running these $\abs{Q}^2$ queries in parallel,
%%     it follows that $\mathcal{G}$ can be constructed in $\NC^2$.
%%     \philip{this seems technically wrong, since constructing is not a decision problem}
%%     \michael{Some people call the functional variant as $\mathsf{FNC}^i$, but it's not really standard I think. We could simply write in the preliminaries that by abuse of language we use $\NC^2$ to refer to both the decision/functional variants.}
%% \end{proof}

\begin{theorem}
  The reachability problem for COCA with equality tests is in $\NC^2$.
\end{theorem}

\begin{proof}
  Let us first argue that the graph $\mathcal{G}$ can be constructed
  in $\NC^2$. There are at most $\abs{Q}$ nodes in $\mathcal{G}$, and
  hence at most $\abs{Q}^2$ edges. Note that an edge $(p'(x), q'(y))$
  is present iff $p'(x) \steps{*} q'(y)$ in $\V_{p', q'}$,
  which can be decided in $\NC^2$ by \autoref{thm:nc2},
  as $\V_{p',q'}$ is a COCA. By running
  these $\abs{Q}^2$ queries in parallel, it follows that $\mathcal{G}$
  can be obtained in $\NC^2$.

  Once the graph $\mathcal{G}$ has been constructed, by
  \autoref{lemma:reach-equivalent}, it suffices to test reachability
  from $p(a)$ to $q(b)$ in $\mathcal{G}$. Since graph reachability is
  in $\NL \subseteq \NC^2$, we are done.
\end{proof}

\section{Guarded COCA reachability}
\label{sec:ptime}
%\subsection{General guarded COCA}\label{subsec:ptime}
We now turn to the reachability problem $p(a) \rightarrow_* q(b)$ for a
guarded COCA $\W = (Q, T, \tau )$.

In contrast to COCAs, for guarded COCAs $\Post{p,q}{a}$ does not necessarily
admit a decomposition into a constant number of intervals. Nevertheless, we
show that it can always be decomposed into a linear number of intervals with
respect to the number of states. Throughout this section, we write
$\intervals(R)$ to denote the unique decomposition of a set $R \subseteq \mathbb{Q}$
into maximal disjoint nonempty intervals. For example, $\intervals([3, 4] \cup (4, 5) \cup (5, +\infty)) = \{[3, 5), (5, +\infty)\}$ and
$\intervals(\emptyset) = \emptyset$.

% \begin{proposition}\label{pro:bnd-ints}
%   The number of disjoint intervals in $\intervals(\Post{p,q}{a})$ is at most $4|Q| + 4$.
% \end{proposition}
% %
% We prove something more general which will be of use later in the
% analysis of our algorithm.

\subsection{Controlling the number of intervals}\label{sec:control-ints}
% 
% We define $\intervals_U \subseteq \intervals$ as the \emph{set of unbounded intervals}. Formally,
% $$
% \mathcal{I}_U \defeq \{I \in \intervals \mid \inf I = -\infty \text{ or } \sup I = +\infty\}.
% %\mathcal{I}_U = (-\infty,+\infty) \cup \bigcup_{z \in \Q} \set{(-\infty,z), (-\infty, z], (z,+\infty), [z,+\infty)}.
% $$
We will prove that, for every $k$, the set $\{b \in \mathbb{Q} \mid
p(a) \rightarrow_\rho q(b), |\rho| = k\}$ decomposes into a polynomial number
of intervals. To do so, we will bound the size of the 
decomposition of sets obtained by updating $A =
[a,a]$ with operations which suffice to implement continuous
counter updates and guard tests. More precisely, the following operations: Minkowski sums,
intersections (with elements of $\mathcal{L} = \{\tau(q) \mid q \in Q\}$), and
unions (with sets constructed similarly).
For technical reasons, we also consider a fourth
operation.

\newcommand{\itmcmd}[1]{\makebox[10pt][l]{#1}}

Let us fix a bounded interval $A \in \intervals$ and 
$\mathcal{L} \subseteq \intervals$. We write
$P_{\mathcal{L}} \defeq \set{\inf I, \sup I \mid I \in \mathcal{L}}$
to denote the set of endpoints within $\mathcal{L}$. We define the
\emph{MIUN-closure} (short for Minkowski sum, Intersection, Union, and New), of interval $A$
w.r.t.\ $\mathcal{L}$, as the smallest collection $\mathcal{C} \subseteq 2^\Q$ such that $A \in \mathcal{C}$ and:
\begin{itemize}
    \item \itmcmd{M:} if $B \in \mathcal{C}$ and $z \in \Q_{> 0}$ then $B + (0,z], B + [-z,0) \in \mathcal{C}$;
    \item \itmcmd{I:} if $B \in \mathcal{C}$ and $L \in \mathcal{L}$ then $B \cap        
      L \in \mathcal{C}$;
    \item \itmcmd{U:} if $B,B' \in \mathcal{C}$ then $B \cup B' \in \mathcal{C}$;
    \item \itmcmd{N:} if $B \in \mathcal{C}$ and $I \in \intervals$ s.t.\ $\overline{I} \cap P_{\mathcal{L}} \neq \emptyset$ then $B \cup I \in \mathcal{C}$.
\end{itemize}

The forthcoming lemma forms the basis of our bound. It is based on so-called
\emph{indicator functions} which give us, for every interval $I$, the set of endpoints of $\mathcal{L}$ and $A$ that belong to
the closure of $I$. As we will see later, the set of endpoints needed to
analyse COCA is small. Furthermore, all MIUN-operations are such that sets of
$\mathcal{C}$ decompose into intervals whose closure contains at least one
such endpoint.

More formally, for all $B \in \mathcal{C}$, let $\phi_B \colon \intervals(B)
\to 2^{P_{\mathcal{L}} \cup P_A}$ be the function defined as $\phi_B(I) \defeq \overline{I}
\cap (P_{\mathcal{L}} \cup P_A)$.

\begin{lemma}\label{lem:phi}
%Let $A \in \intervals$ be a bounded interval, where $\overline{A} =
%[a_\ell,a_u]$, and let $\mathcal{L} \subseteq \intervals$. Let $\mathcal{C}$
%be the MIUN-closure of $A$ with respect to $\mathcal{L}$.
%  For all $B \in
%\mathcal{C}$ consider the function $\phi_B \colon \intervals(B) \to
%2^{P_{\mathcal{L}} \cup \{a_\ell,a_u\}}$ defined as $\phi_B(I) \defeq
%\overline{I} \cap (P_{\mathcal{L}} \cup \{a_\ell,a_u\})$. 
    We have $\phi_B(I) \neq
\emptyset$ for all $B \in \mathcal{C}$ and $I \in \intervals(B)$.
\end{lemma}

\begin{proof}%[Proof of \autoref{lem:phi}]
    We proceed by induction on the definition of MIUN-closures. We 
    define $\mathcal{C}_0 \defeq
    \set{A}$ and $\mathcal{C}_{i+1}$ as $\mathcal{C}_i$ extended with all sets obtained
    by applying the MIUN-operations applied to any $B, B' \in \mathcal{C}_i$. We will show
    that the lemma holds for all $\mathcal{C}_i$, which will conclude the
    proof since $\mathcal{C} = \bigcup_{i\in \N}\mathcal{C}_i$.

    We have $\intervals(A) =
    \{A\}$ and the claim holds since $P_A \subseteq \phi_A(A)$.  For the
    induction step, we suppose the claim holds for $\mathcal{C}_i$. We have to
    prove that for all $C \in \mathcal{C}_{i+1}$ and all $I \in \intervals(C)$
    it holds that $\phi_C(I) \neq \emptyset$. Notice that this is trivial if $C$ is obtained
    from $B \in \mathcal{C}_i$ by application of the New operation.
              
    First, we consider the Minkowski sum. Consider some $B \in \mathcal{C}_i$
    with the function $\phi_B$ and let $I \in \set{(0, z], [-z,0)}$ for some
    $z \in \Q_{>0}$. Let $C \defeq B + I$. For all $J \in \intervals(C)$
    there exists $K_B \in \intervals(B)$ such that $K_B \subseteq J$. Thus,
    $\phi_B(K_B) \subseteq \phi_C(J)$ and the claim holds by the inductive
    hypothesis for $\phi_B$.
      
    Second, we consider intersections. We only deal with intervals of
    the form $[\ell,+\infty)$, $(\ell,+\infty)$, $(-\infty,\ell)$, or
    $(-\infty,\ell]$, since intersection with any interval can be expressed by
    at most two consecutive intersections with intervals of this form. Let $B \in
    \mathcal{C}_i$ and $L \in \mathcal{L}$.  Suppose that $\overline{L} =
    [\ell,+\infty)$ and let $C \defeq B \cap L$.  Recall that $\ell \in
    P_{\mathcal{L}}$. Observe
    that $\intervals(B)$ contains at most one interval $I$ such that $\ell \in
    \overline{I}$. If such an $I$ exists, then $\ell \in \phi_C(\overline{I
    \cap L})$. For all other intervals $J \in \intervals(B)$, we have that $J
    \cap L$ is either $J$ or $\emptyset$. If the intersection is nonempty,
    then $\phi_C(J) = \phi_B(J)$ and the claim holds by inductive
    hypothesis. If $\cl{L}$ is instead of the form
    $(-\infty,\ell]$, then we proceed similarly.
          
    Finally, we consider unions. Let $B,B' \in \mathcal{C}_i$ and 
    $I \in \intervals(B \cup B')$. By definition, there exists $J \in \intervals(B)
    \cup \intervals(B')$ with $J \subseteq I$. Therefore, either $\phi_B(J)$ or
    $\phi_{B'}(J)$ is nonempty and contained in $\phi_{B \cup B'}(I)$.
\end{proof}

% Let $\mathcal{D}$ be a collection of subsets of rational numbers. We define the union closure of $\mathcal{D}$ as the collection $\mathcal{U}(\mathcal{D})$, where:
% if $\mathcal{E} \subseteq \mathcal{D}$
% then $\bigcup \mathcal{E} \in \mathcal{U}(\mathcal{D})$. Notice that $\mathcal{E}$ does not need to be a finite subset.
% 
% \begin{lemma}\label{lemma:union}
% Let $A \in \intervals$ be a bounded interval and let $\mathcal{L} \subseteq \intervals_U$. Let $\mathcal{C}$ be the MI-closure of $A$ with respect to $\mathcal{L}$.
%   For all $\mathcal{D} \in \mathcal{U}(\mathcal{C})$
%   there is a function $\phi_\mathcal{D} : \intervals(\mathcal{D}) \to P \cup \{a_\ell,a_u\}$ such
%   that $\phi(B) \in \overline{B}$ for all $B \in \intervals(\mathcal{D})$.
% \end{lemma}

%\begin{lemma}\label{lemma:nothreesome}
%Let $I_1,I_2,I_3 \in \intervals(B)$
%%be pairwise disjoint intervals
%for some set $B \subseteq \Q$.
%There is no point $z \in \Q$ such that $z \in \overline{I_1} \cap \overline{I_2} \cap \overline{I_3}$.
%\end{lemma}

\begin{lemma}\label{lemma:nothreesome}
  For every set $B \subseteq \Q$ and every pairwise distinct intervals
  $I_1, I_2, I_3 \in \intervals(B)$, it is the case that $\cl{I_1}
  \cap \cl{I_2} \cap \cl{I_3} = \emptyset$.
\end{lemma}

A point can belong to at most one interval among disjoint
intervals. Moreover, a point
can belong to at most two closures, \eg consider $[0,1)$ and $(1,2]$. This is no
longer possible for three intervals due to maximality of intervals in
$\intervals(B)$. Thus, the proof of \autoref{lemma:nothreesome} follows from a
simple case analysis.

Now, we show that if $\mathcal{L}$ is finite, then there is a polynomial bound on the number of intervals within the
decomposition of any set from the MIUN-closure $\mathcal{C}$. More formally:

\begin{lemma} \label{lemma:boundedunion}
%Let $A \in \intervals$ be a bounded interval and let $\mathcal{L} \subseteq \intervals$ be a finite set of intervals.
%    Let $\mathcal{C}$ be the MIUN-closure of $A$
%with respect to $\mathcal{L}$.
  If $\mathcal{L}$ is finite, then
$\intervals(B)$ consists of at most $4(|\mathcal{L}| + 1)$ intervals, for every $B \in \mathcal{C}$
\end{lemma}

% Notice that this in particular proves that $\intervals(B)$ is finite for all $B \in \mathcal{C}$. Then $B + (0,z], B + [-z,0) \in \mathcal{C} \in \mathcal{C}$ for all $z \in \Q_{>0}$.

\begin{proof}
  By \autoref{lemma:nothreesome}, there are at most two pairwise disjoint intervals that share a point in their
  closure. By \autoref{lem:phi}, the indicator function guarantees that
  $\overline{J} \cap (P_{\mathcal{L}} \cup P_A) \neq \emptyset$ for
  all $J \in \intervals(B)$. Thus, $\intervals(B)$ has at most
  $2(2|\mathcal{L}| + 2)$ intervals. Otherwise, by the pigeonhole principle, a point of $P_{\mathcal{L}} \cup P_A$ would belong to at
  least three closures of intervals from $\intervals(B)$.
\end{proof}

% \begin{proof}[Proof of Proposition~\ref{pro:bnd-ints}]
% Consider a path $\pi \in \Paths{p}{q}$. We show that
% $\Post{\pi}{a}$ is an element of the MI-closure of $A = [a,a]$, with respect to
% $\mathcal{L} = \set{\topguard{q}, \botguard{q} \mid q \in Q}$.
% Indeed, the Minkowski sum allows to
% adding a scaled transition \eg $(0,z]$ for a transition $(r,z,r') \in T$. The intersections
% allow to express guard constraints from $\tau$.
% 
% By definition $\Post{p,q}{a}$ is the union of
% $\Post{\pi}{a}$ over all $\pi \in \Paths{p}{q}$. Since $\mathcal{L} \leq 2|Q|$,
% Proposition~\ref{pro:bnd-ints} follows from Lemma~\ref{lemma:boundedunion}.
% \end{proof}

% Before we present an algorithm to compute $\Post{p,r}{a}$ for all $r \in Q$,
% we give an inductive characterization of it based on runs of bounded length.

\subsection{Approximations of the reachability function}

It will be convenient to
manipulate mappings from states to (under-approxi\-ma\-tions of) their
reachability functions.
% $\Post{q,r}{a}$ from $p(a)$.
% Hence, we define the set
% $\mathcal{R} := \{ R : Q \to 2^\mathbb{Q}\}$ of all mappings from states to
% subsets of rational numbers. We further define a partial order on $\mathcal{R}$
% by setting
% \[
%     R \preceq R' \iffdef R(q) \subseteq R'(q) \text{ for all } q \in Q.
% \]
% Note that the structure $(\mathcal{R}, \preceq)$ is a complete lattice.
% We say a function $g : \mathcal{R} \to \mathcal{R}$ is $\preceq$-monotone if
% for all $x, y$ in the domain of $g$ it holds that  $x\preceq y$ implies $g(x)
% \preceq g(y)$.Finally, we define $R_\emptyset$ to be the empty mapping $q
% \mapsto \emptyset$; and $R_0$ to be the start
% mapping, \ie $R_0(p) = \set{a}$ and $R_0(q) = \emptyset$ for all $q \in Q
% \setminus \{p\}$.
We consider the mappings 
$\mathcal{R}_Q \defeq \{ R \colon Q \to 2^\mathbb{Q}\}$.
An example of such a mapping is $\reach_{p(a)}$ defined as
$\reach_{p(a)}(q) \defeq \Post{p,q}{a}$.
Given $R,R' \in \mathcal{R}_Q$, we write $R \preceq R'$ iff $R(q) \subseteq R'(q)$
for all $q \in Q$.
We seek to define a sequence of mappings $R_0 \preceq R_1 \preceq \cdots$ such that $R_n = \reach_{p(a)}$ for some $n \in \N$.

%\michael{Can $R_n \subset R_{n+1}$, i.e.\ can we go beyond?}
%\tim{We don't really care here, if I'm not mistaken. By applying either succ or acc we can't, but this is handled in later lemma's.} 
% We use mappings from $\mathcal{R}$ to represent
% under-approximations of reachable values for all states in $Q$. That is,
% under-approximations of $\reach_{p(a)} \in \mathcal{R}$ defined as
% $\reach_{p(a)} : q \mapsto \Post{p,q}{a}$.
% Below, we introduce a function to
% update mappings to tighten the approximations by symbolically
% taking transitions in a forward fashion.

For all state $q \in Q$, we define the \emph{successor mapping-update function}
$\TheBetterPost_{q}{} \colon \mathcal{R}_Q
\to 2^{\mathbb{Q}}$ as follows:
\begin{alignat*}{2}
  \TheBetterPost_{q}(R) & {} \defeq R(q)\\
    &\cup \bigcup\{(R(r) + (0,z]) \cap \tau(q) \mid (r,z,q) \in T, z > 0\}\\
    &\cup \bigcup\{(R(r) + [z,0)) \cap \tau(q) \mid (r,z,q) \in T, z < 0\} \\
    &\cup \bigcup\{R(r) \cap \tau(q) \mid (r,0,q) \in T\}.
\end{alignat*}

%----
% We then define $\TheBetterPost:\mathcal{R} \rightarrow \mathcal{R}$ for all $R
% \in \mathcal{R}$ as $\TheBetterPost(R) \defeq q \mapsto
% \TheBetterPost_{q}(R)$. Further, for all $n \in \mathbb{N}$, we recursively
% define the $n$-fold composition $\TheBetterPost^n$ of $\TheBetterPost$ by
% letting $\TheBetterPost^0(R) \defeq R$ and $\TheBetterPost^{n}(R) \defeq
% \TheBetterPost(\TheBetterPost^{n-1}(R)))$ for all $n > 0$. 
Let $\TheBetterPost \colon \mathcal{R}_Q \to \mathcal{R}_Q$
be defined as $\TheBetterPost(R)(q) \defeq \TheBetterPost_{q}(R)$.
Below, we state the key property enjoyed by $\TheBetterPost$. In words, its
$i$-fold composition coincides with the set of configurations reachable
via runs of length at most $i$. It can be easily proved by induction on the
definition of $\TheBetterPost$.
\begin{restatable}{lemma}{postSucks}
\label{lem:post-succs}
Let $S_0 \in \mathcal{R}_Q$ and $S_i \defeq
\TheBetterPost(S_{i-1})$ for all $i \ge 1$. The following holds:
\[
  S_i(q) = 
    \bigcup_{p \in Q} \{b \in \mathbb{Q} \mid a \in S_0(p),
    p(a) \rightarrow_{\rho} q(b) \text{ and } |\rho| \leq i\}.
\]
\end{restatable}
%\philip{suggestion: the previous notation has large vertical space on the right-hand side.
%maybe this might look better?}
%
%\begin{gather*}
%  S_i(q) \; = \; \bigcup_{p \in Q}
%  \left\{
%\begin{array}{l}
%  b \in \mathbb{Q} \mid a \in S_0(p), \text{ and} \\
%  p(a) \rightarrow_{\rho} q(b),\ \rho \text{ has length at most } i
%\end{array}
%\right\}
%\end{gather*}
%
%\philip{not sure if it looks better or worse, but it saves space}

Now we can state a proposition that shows how the previous section relates to
these definitions. Let us fix a
configuration $p(a)$. We will focus on the MIUN-closure $\mathcal{C}$
%is the MIUN-closure of
%$\W$ and $p(a)$ if it is the MIUN-closure
of $A \defeq [a,a]$ with respect
to $\mathcal{L} \defeq \set{\guard{q} \mid q \in Q}$.
We say that a mapping $S
\in \mathcal{R}_Q$ is \emph{$\mathcal{C}$-valid} if $S(q) \in \mathcal{C}$ for
all $q \in Q$.

\begin{proposition}\label{proposition:succ}
%Let $\mathcal{C}$ be the MIUN-closure of $\W = (Q, T, \tau)$ and $p(a)$ and
Let $S \in \mathcal{R}_Q$ be a $\mathcal{C}$-valid mapping. We have $S \preceq
\TheBetterPost(S)$ and $\TheBetterPost(S)$ is a $\mathcal{C}$-valid mapping.
\end{proposition}

\begin{proof}
We have $S \preceq \TheBetterPost(S)$ directly from the definition of
$\TheBetterPost$. To prove that $\TheBetterPost(S)(q) \in \mathcal{C}$, it
suffices to observe that $\TheBetterPost_q(S)$ is defined using Minkowski
sums, intersections, and unions which are building blocks of MIUN-closures.
\end{proof}

\subsection{Accelerations}

Unfortunately, applying $\TheBetterPost$ might not give us $\reach_{p(a)}$ in a
small or even finite number of steps, \eg\ if $\reach_{p(a)}(q)$ is
unbounded for some $q \in Q$. We introduce another operation on mappings
to resolve this. We start by defining some special form of cycles.

Let us fix a mapping $S_0 \in \mathcal{R}_Q$ and let $S_{i+1} \defeq
\TheBetterPost(S_i)$ for every $i \geq 0$.
%Consider a run $\rho = \alpha_1t_1
%\dots \alpha_n t_n$.
We say that a run $\rho = \alpha_1t_1 \cdots \alpha_n t_n$ is a \emph{positively expanding
cycle} from $S_0$ if it is admissible and there exist configurations $p_0(a_0), p_1(a_1), \dots, p_n(a_n)$ such that:
\begin{enumerate}[label=(\arabic*)]
 \item $p_0 = p_n$ and $\Effect{\rho} > 0$;
  \item $a_0 \in S_0(p_0)$
    %$\rho$ is admissible from $p_0(a_0)$
    and $p_0(a_0) \steps{\rho_i} p_i(a_i)$ for all $i \ge 1$; and

 \item $a_i \in S_i(p_i) \setminus S_{i-1}(p_i)$ for all $i \ge 1$.
\end{enumerate}
Moreover, letting $I_0,\dots,I_n$ be the sequence of intervals such that $a_i \in
I_i \in \intervals(S_i(p_i))$ for all $i \in \set{0,\dots,n}$, we require:
\begin{enumerate}[label=(\arabic*)]
 \setcounter{enumi}{3}
 \item $I_0 \subseteq I_n$;

 \item for all $i \ge 1$, there is a unique interval $I_i' \in
   \intervals(S_{i-1}(p_{i}))$ such that $I_{i}' \subseteq I_{i}$; and

 \item $a_i \ge \sup (I_{i}')$ for every $i \ge 1$.
\end{enumerate}
Intuitively, the third condition states that each $a_i$ is a ``new'' value, and the
fifth and sixth conditions state that $a_i$ expands some interval towards the
top.

Similarly, we say that $\rho$ is a \emph{negatively expanding cycle}
from $S_0$ if in the first item we replace $\Effect{\rho} > 0$ with
$\Effect{\rho} < 0$; and in the last item we replace $a_i \ge \sup (I_{i}')$
with $a_i \le \inf (I_{i}')$.
%\philip{for positive cycles, you write \emph{positively expanding cycle} from $S_0$,
%but for negative cycles you write \emph{negatively expanding cycle} from $p(a)$.
%I assume the first notation is right?}

The following property follows from the definitions:

\begin{lemma}\label{lem:same-int}
  It holds that $a_0,a_n \in I_n \subseteq \reach_{p_0(a_0)}(p_0)$.
\end{lemma}

It transpires that the $\TheBetterPost$ function always yields expanding
cycles after a polynomial number of applications. The proof of this claim
relies on our bounds for interval decompositions of sets from the
MIUN-closure. The computational part is a simple backward construction of a
run containing a cycle. A full proof is given in \autoref{sec:proof-cycle} of the appendix.
%\philip{suggestion: ref to last subsection, so there is a clickable link}

%\michael{Do we need to define $S_0$ and $S_{i+1}$ again below?}
\begin{restatable}{proposition}{findCycle}\label{proposition:cycle}
%Let $\mathcal{C}$ be the MIUN-closure of $\mathcal{V} = (Q, T, \tau)$ and $p(a)$.
%Let $S_0 \in \mathcal{R}_Q$ be $\mathcal{C}$-valid and let $S_{i+1} \defeq
%\TheBetterPost(S_i)$ for all $i \in \N$.
For some $n$ polynomially
bounded in $|Q|$, at least one of the following holds: $S_n = S_{n+1}$,
\begin{itemize}
 \item there is a positively expanding cycle $\rho$ from $S_n$, or
 \item there is a negatively expanding cycle $\rho$ from $S_n$.
\end{itemize}
Moreover, it can be determined in time $|Q|^{\Oh(1)}$ whether the second
or third case hold, and then $\rho$ and its witnessing configurations can be
computed in time $|Q|^{\Oh(1)}$.
\end{restatable}

We are ready to define the \emph{acceleration} operation. 
Let $\rho$ be a positively or negatively expanding cycle from $S \in
\mathcal{R}_Q$ and let $p_0(a_0), p_1(a_1), \dots, p_n(a_n)$ be the
configurations witnessing the run. 
%Note that $p_0(a_0), p_1(a_1), \dots, p_n(a_n)$ is the suffix of an admissible run from the initial configuration.
Let $I_0,\dots, I_n$ be the intervals given by
the definition of expanding cycles. If $\rho$ is
positively expanding, then we define $\delta^+_i \defeq \sup \guard{p_i} - a_i$ for all
$i \in \set{1,\dots, n}$. If $\rho$ is a negatively expanding
cycle, then we define $\delta^-_i \defeq a_i - \inf \guard{p_i}$. Let $j \in
\set{1,\dots,n}$ be such that:
\[
\delta^+_j = \min\set{\delta^+_i \mid 1 \leq i \leq n} \text{ or } \delta^-_j = \min\set{\delta^-_i \mid 1 \leq i \leq n}.
\]
% Let $I_j' \in \intervals()$
We define $\Acceleration$ so that, given $\rho$ and the mapping $S$, it outputs a
new mapping $\Acceleration(S,\rho) = S'$. If $\rho$ is positively expanding
from $S$ then $S'(q) \defeq
S(q)$ for all $q \neq p_j$ and:
\[
  S'(p_j) \defeq S(p_j) \cup \underbrace{I_j \cup \left( \guard{p_j} \cap
  [a_j,+\infty)\right)}_{{} = K \subseteq \mathbb{Q}}.
\]
%\begin{itemize}
% \item if $\delta^+ = + \infty$ then $S'(p_j) = S(p_j) \cup I_j \cup  [a_j, +\infty)$;
% \item otherwise let $g = \sup(\guard{p_j})$; if $g = a_j$ then $S'(p_j) = S(p_j) \cup I_j \cup [a_j,g]$, and otherwise $S'(p_j) = S(p_j) \cup I_j \cup [a_j,g)$.
%\end{itemize}
Recall that $a_j \in I_j$, and $a_j \in \guard{p_j}$ since $a_0 \rightarrow_{\rho[1..j]} a_j$, so $K$ is an interval.
%\tim{I got lost in the notation, you define $\delta^+_i$ and $\delta^+_j$, but in the definition of acceleration, you write about $\delta^+$.}
%\filip{That's a typo should be $\delta^+_j$}
%\philip{I don't get this; it might be that we add new values, but no new interval, right?
%imagine $I_j$ is strictly larger than $I'_j$, then $I'_j \cap I_j \neq \emptyset$
%and $I_j \cap [a_j, g) \neq \emptyset$, so no additional interval is added, but $I'_j$ is expanded
%to $g$. Or maybe this is what you mean: that the new points form an interval, that is no necessarily disjoint from what intervals already exist}
Also, since $\rho$ is positively expanding and $j \geq 1$, we have $a_j \not\in
S(p_j)$ and $S'(p_j) \setminus S(p_j) \neq \emptyset$.
Similarly, if $\rho$ is negatively expanding then $S'(q) \defeq S(q)$ for all $q \neq p_j$ and:
\[
  S'(p_j) \defeq S(p_j) \cup I_j \cup \left((-\infty,a_j] \cap \guard{p_j} \right).
\]

%\begin{itemize}
% \item if $\delta^- = + \infty$ then $S'(p_j) = S(p_j) \cup (-\infty,a_j] \cup I_j$;
% \item otherwise let $g = \inf(\guard{p_j})$; if $a_j = g$ then $S'(p_j) = S(p_j) \cup [g,a_j] \cup I_j$, and otherwise $S'(p_j) = S(p_j) \cup (g,a_j] \cup I_j$.
%\end{itemize}

%
%\begin{lemma}
%  Let $a \in \Q$, and let $\rho$ be a cycle,
%  and $\theta$ be a run with $\upath{\theta} = \rho$
%  which is admissible from $a$.
%  Further, let $b \in \enab{\rho}$ such that $b \steps{\theta'} a$
%  for some $\theta'$ with $\upath{\theta'} = \rho$.
%
%  If $\Effect{\theta} > 0$,
%  then for all $b' \in a + \Effect{\theta} + (0,\Effect{\theta})$,
%  it holds that $a \steps{\rho} b'$.
%
%  If $\Effect{\theta} < 0$,
%  then for all $b' \in a + \Effect{\theta} + (\Effect{\theta},0)$,
%  it holds that $a \steps{\rho} b'$.
%\end{lemma}
%\begin{proof}
%  Let us show only the case where $\Effect{\theta} > 0$.
%  The other case is symmetric.
%
%  Let us fix an arbitrary $b' \in a + \Effect{\theta} + (0, \Effect{\theta})$.
%  In the course of the proof, we will construct a run $\sigma$
%  with $\upath{\sigma} = \rho$,
%  which is admissible from $a$ and such that $a + \Effect{\sigma} = b'$
%  for a fixed 
%
%  Clearly there exists $\beta \in (0,1)$ such that
%  $a + \Effect{\beta \theta} = b'$.
%  If $\beta \theta$ is admissible from $a$, we are done, since
%  $\beta \theta$ is an admissible run from $a$ to $b'$.
%\end{proof}
%\tim{left off here}

\begin{lemma}\label{lemma:acceleration}
%Let $\mathcal{C}$ be the MIUN-closure of $\mathcal{V} = (Q, T, \tau)$ and $p(a)$.
  Let $S \in \mathcal{R}_Q$ be a $\mathcal{C}$-valid mapping such that $S
  \preceq \reach_{p(a)}$ and let $\rho$ be a positively or negatively expanding cycle from $S$.
  If $S' = \Acceleration(S,\rho)$, then $S \preceq S'$, $S'$ is a $\mathcal{C}$-valid mapping, and $S' \preceq \reach_{p(a)}$.
\end{lemma}

\begin{proof}
We have $S \preceq S'$ directly from the definition of $\Acceleration$. Similarly, $S'$ is $\mathcal{C}$-valid because the operation to define $S'(p_j)$ is the ``New'' operation since the closure of the added interval always contains one of the endpoints from $\guard{p_j}$. It remains to prove that $S' \preceq \reach_{p(a)}$.

We assume that $\Effect{\rho} > 0$; the proof is similar for the other case.
Let $\rho = \alpha_1t_1 \cdots \alpha_nt_n$. Let $p_0(a_0), p_1(a_1), \dots,
p_n(a_n)$ and $I_0,\dots, I_n$ be the configurations and intervals given by the
definition of positively expanding cycles. Let $j$ be the index minimising
$\delta^+_j$. We must prove that $S'(p_j) \preceq \reach_{p(a)}$. Since $S
\preceq \reach_{p(a)}$ and $a_0 \in S(p_0)$ by definition of positively
expanding cycles, it suffices to show that for every $b \in S'(p_j) \setminus
S(p_j)$ there is an admissible run from $p_0(a_0)$ to $p_j(b)$.

If $\delta^+_j = +\infty$, then $\sup \guard{p_i} = +\infty$ for all $i \in
\set{1, \dots, n}$. Since $\Effect{\rho} > 0$, %the following holds
for all $\alpha, \beta
\in (0,1]$ and $m \in \N$ the run $\rho' \defeq (\beta \rho)^{m}
\alpha \FromTo{\rho}{1}{j}$ is admissible from any $p_n(a')$ with $a' \geq
a_n$, to state $p_j$. Note that $\Effect{\rho'}
= m\beta\Effect{\rho} + \alpha\Effect{\rho[1..j]}$, which can be any positive
rational number by properly choosing $\alpha$, $\beta$ and $m$. Thus, $b \in
\reach_{p(a)}(p_j)$ for every $b > a_n$.

It remains to consider the case 
$a_j \leq a_n$ to prove the claim for every $b \in
[a_j,a_n]$. Let $\epsilon \in (0, a_n - a_0]$.
Since $a_0 < a_0 +
\epsilon \leq a_n$, we have $a_0 + \epsilon \in I_n \subseteq \reach_{p(a)}(p_0)$ where the latter follows from
\autoref{lem:same-int}.
Recall that $\rho'$ is admissible from all $p_n(a')$ with $a'
\geq a_n$, and hence from $a_n + \epsilon$. Thus,
$p(a) \rightarrow_* p_0(a_0 + \epsilon) \rightarrow_{\rho} p_n(a_n +
\epsilon)$, and $p_0(a_0 + \epsilon) \rightarrow_{\rho[1..j]} p_j(a_j +
\epsilon)$. This shows that $b \in \reach_{p(a)}(p_j)$
for all $b \in [a_j, a_n]$.

%\tim{I think this run is admissible from $p_n(a_n)$ since by scaling the initial run through the loop can violate some lower-bounds. The fact that $[a_0, a_n]$ is reachable comes from the fact that they are in the same interval.} 
%\filip{you're right in both cases}

Now, suppose $\delta^+_j < +\infty$. If $a_j = \sup \guard{p_j}$, then
we are done because $a_j \in I_j \subseteq
\reach_{p(a)}(p_j)$. Otherwise, let $b \in [a_j,+\infty) \cap \tau(p_j)$. We
need to prove that $b \in \reach_{p(a)}(p_j)$. Note that, by definition, we have $0
\le b - a_j \le \delta^+_j$.

%\tim{We're missing the case where $g$ can be reached, but $g > a_j$. So when b = g.} 
% We consider two cases. First, suppose that $\Effect{\rho} \ge b - a_j$. We show that $p_0(a_0 + b-a_j) \steps{\rho_j} p_j(b)$ is an admissible run. Indeed, recall that $p_0(a_0) \steps{\rho} p_n(a_n)$ is admissible, thus $a_i + b - a_j \ge a_i \ge \inf(\guard{p_i})$. And since $b - a_j < \delta^+_j$ we have $a_i + b - a_j \le a_i + \delta^+_j \le a_i + \delta^+_i = \sup(\guard{p_i})$. It remains to show that $a_0 + b-a_j \in \reach_{p(a)}(p_0)$. Recall that $p_0 = p_n$ and that by definition of the positively expanding cycle $a_0, a_n \in I_n$. This concludes the proof as $I_n \subseteq \reach_{p(a)}(p_0)$, $I_n$ is an interval, and $a_0 \le a_0 + b-a_j \le a_0 + \Effect{\rho} = a_n$.
% 
% Now, suppose that $\Effect{\rho} < b - a_j$.
Let $m \in \N$ and $c \in \Q_{\ge 0}$ be the unique numbers that satisfy $b - a_j
= m \cdot \Effect{\rho} + c$ and $c < \Effect{\rho}$. Since $a_0 \le a_0 + c
\le a_0 + \Effect{\rho} = a_n$, then by \autoref{lem:same-int}
%$I_n \subseteq \reach_{p(a)}(p_0)$, and $a_0,
%a_n \in I_n$ (by definition of positively expanding cycles) 
we conclude that
$a_0 + c \in \reach_{p(a)}(p_0)$. Notice that $a_j + c + m\cdot\Effect{\rho} =
b$. It thus remains to prove that $p_0(a_0 +c)
\steps{\rho^{m}\rho[1..j]}p_j(b)$.
%is an admissible run.
We prove something stronger, namely that $p_0(a_0 +c)
\steps{\rho^{m+1}}p_n(a_n + b - a_j)$.
% is an admissible run.
%\michael{The arrow notation means it is admissible.}
Since
$\Effect{\rho} > 0$, for the bottom guards it suffices to check whether the
configurations are large enough when $\rho$ is applied the first time. Indeed,
since $\rho$ is admissible from $p_0(a_0)$, we get $a_i + c + \Effect{\rho_i}
\ge a_i + \Effect{\rho_i} \ge \inf \guard{q_i}$. Similarly, for the top
guards, it suffices to check whether the configurations are small enough when $\rho$ is
applied last.

Indeed, since $b - a_j \leq \delta^+_j$, we have $a_i + c +
m\cdot \Effect{\rho} = a_i + b - a_j \leq a_i + \delta^+_j \le a_i + \delta^+_i =
\sup \guard{p_i}$. If $\sup \guard{p_j} \not\in \guard{p_j}$, then $b - a_j <
\delta^+_j$ and the previous inequalities are strict.
\end{proof}

\subsection{Polynomial time algorithm}

We summarise how to obtain the polynomial time algorithm
for deciding $p(a) \steps{*} q(b)$. We begin with the mapping
$R_0 \in \mathcal{R}_Q$ defined as $R_0(p) \defeq [a,a]$ and $R_0(r) \defeq
\emptyset$ for every $r \neq p$. Clearly, $R_0 \preceq \reach_{p(a)}$. The next
mappings $R_{1}, R_2, \dots $ are defined as follows. Suppose we have defined
$R_0,\dots, R_i$. Let $S_0^i \defeq R_i$ and $S_{j+1}^i \defeq \TheBetterPost(S_j^i)$
for all $j \ge 0$.

By \autoref{proposition:cycle}, we will either find
an expanding cycle $\rho$ from some $S_n^j$, where
$n$ is bounded polynomially, or we will find some $S_n^j = S_{n+1}^j$, again
for $n$ bounded polynomially.
%In the latter case, we define $R_{i+j} \defeq
%S_j^i$ for $j \in \{1, \ldots, n\}$. We have $R_i \preceq \reach_{p(a)}$ for all
%defined $R_i$, and hence are done by \autoref{lem:post-succs}.
If there is an expanding cycle --- a fact which, by \autoref{proposition:cycle},
we can check in polynomial time --- then we define $R_{i +
j} \defeq S_j^i$ for $1 \le j < n$ and $R_{i + n} \defeq \Acceleration(R_{i +
n-1},\rho)$.
Otherwise, we define $R_{i+j} \defeq S_j^i$ for $j \in \{1,\dots,n\}$ and
the algorithm returns $R_{i+n}$. By \autoref{lem:post-succs} and
\autoref{lemma:acceleration}, we have $R_{i} \preceq \reach_{p(a)}$ for all defined
$R_i$. Hence, if the algorithm terminates then by \autoref{lem:post-succs} it
returns $\reach_{p(a)}$.

%\guillermo{We are missing an argument for correctness right? polynomial
%termination is argued for below}

The rest of this section is devoted to proving that the above-described
algorithm has a polynomial worst-case running time. It suffices to argue that
expanding cycles can only be found some polynomial number of times.

\begin{proposition}\label{prop:algo:ptime}
  The algorithm computes a representation of $\reach_{p(a)}$ in time
  $|Q|^{\Oh(1)}$.
\end{proposition}

%We prove that such $N$ exists and that it is bounded polynomially, which will conclude the proof.
By \autoref{proposition:cycle}, it suffices to show that $\Acceleration$ can be~applied at most polynomially many times. To do so, we introduce the notion of
progressing extensions.
Let $\mathcal{C}$ be the MIUN-closure of $A$ w.r.t.\
$\mathcal{L}$, and let $B,B' \in \mathcal{C}$ be such that $B \subseteq
B'$. We say that $B'$ is a \emph{progressing extension} of $B$ if:
\begin{enumerate}[label=(\arabic*)]
 \item there is $I' \in \intervals(B')$ such that $B \cap I' =
   \emptyset$;\label{itm:progr:a}
\end{enumerate}
or if there are $I \in \intervals(B)$ and $I' \in \intervals(B')$ such that $I \subseteq I'$ and at least one of the following holds:
\begin{enumerate}[label=(\arabic*)]
  \setcounter{enumi}{1}
  \item either $\phi_{B'}(I') \setminus \phi_{B} (I) \neq \emptyset$, or\label{itm:progr:b}

  \item there exists $\ell \in \phi_B(I)$ such that $\ell \not \in I$ and
   $\ell \in I'$.\label{itm:progr:c}
\end{enumerate}
See \autoref{fig:prog-extensions} for a pictorial description of progressing
extensions. Observe that 
in case~\ref{itm:progr:c} we necessarily have that $\ell \in \overline{I}$.

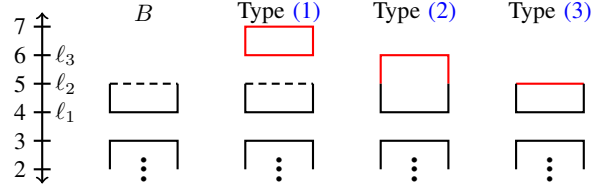
\begin{figure}
  \begin{center}
    \begin{tikzpicture}[thick, y=12pt, transform shape, scale=0.9]
  %% y-axis
  \draw[<->] (0, 2 - 0.5) -- (0, 7 + 0.5);

  \foreach \y in {2,...,7} {
    \draw (-0.125, \y) -- (0.125, \y);
    \node at (-0.325, \y) {\y};
  }

  \node at (0.35, 6) {$\ell_3$};
  \node at (0.35, 5) {$\ell_2$};
  \node at (0.35, 4) {$\ell_1$};

  %% B
  \node at (1.5, 7.5) {$B$};

  \draw (1, 5) -- (1, 4) -- (2, 4) -- (2, 5);
  \draw[densely dashed] (1, 5) -- (2, 5);

  \draw (1, 2) -- (1, 3) -- (2, 3) -- (2, 2);
  \node[font=\huge] at (1.5, 2.25) {$\vdots$};

  %% Type I
  \node at (3.5, 7.5) {Type~\ref{itm:progr:a}};
  
  \draw[red] (3, 6) -- (3, 7) -- (4, 7) -- (4, 6) -- cycle;

  \draw (3, 5) -- (3, 4) -- (4, 4) -- (4, 5);
  \draw[densely dashed] (3, 5) -- (4, 5);

  \draw (3, 2) -- (3, 3) -- (4, 3) -- (4, 2);
  \node[font=\huge] at (3.5, 2.25) {$\vdots$};

  %% Type II
  \node at (5.5, 7.5) {Type~\ref{itm:progr:b}};
  
  \draw[red] (5, 5) -- (5, 6) -- (6, 6) -- (6, 5);

  \draw (5, 5) -- (5, 4) -- (6, 4) -- (6, 5);

  \draw (5, 2) -- (5, 3) -- (6, 3) -- (6, 2);
  \node[font=\huge] at (5.5, 2.25) {$\vdots$};

  %% Type III
  \node at (7.5, 7.5) {Type~\ref{itm:progr:c}};
  
  \draw[red] (7, 5) -- (8, 5);
  \draw (7, 5) -- (7, 4) -- (8, 4) -- (8, 5);

  \draw (7, 2) -- (7, 3) -- (8, 3) -- (8, 2);
  \node[font=\huge] at (7.5, 2.25) {$\vdots$};
\end{tikzpicture}
  \end{center}
\caption{  
  Left: A set $B$ such that $\mathcal{I}(B) = \{(-\infty, 3], [4,5)\}$.
  Right: Example of the three possible types of progressing extensions of $B$.
  Dashed lines denote open interval borders; 
  $\ell_1 = 4, \ell_2 = 5$ and $\ell_3 = 6$ denote 
  values in $P_{\mathcal{L}} \cup P_A$.}
  \label{fig:prog-extensions}
\end{figure}

We show that accelerating leads to a progressing extension. Then, in the
sequel, we provide a polynomial bound on the number of progressing extensions in
a $\subseteq$-increasing sequence.

\begin{lemma}
  Let $\rho$ be an expanding cycle from $R$ and let
  $R' \defeq \Acceleration(R,\rho)$. There is some state $p_j$ such that $R(p_j) \subseteq R'(p_j)$ is a
  progressing extension.
\end{lemma}

\begin{proof}
  Let $p_j \in Q$ be such that $R'(p_j) = R(p_j) \cup I_j \cup J$
  for some interval $J$. In the proof, we write $\alpha_i t_i$, $p_i(a_i)$,
  and $I_i$, as in the definition of expanding cycles. We will assume that $\rho$ is
  positively expanding;
  the other case is similar. Thus, $J \in \set{[a_j,g], [a_j,g),
  [a_j,+\infty)}$, where $g \defeq \sup \guard{p_j}$. Recall that by definition
  $a_j \not \in R(p_j)$ and $a_j \in I_j \in
  \intervals(\TheBetterPost^k(R)(p_j))$ for some $k$. Thus, $I_j \cup J$ is an
  interval.

  If $(I_j \cup J) \cap R(p_j) = \emptyset$, then $R'(p_j)$ is a progressing
  extension due to~\ref{itm:progr:a}. For the remaining case, let $b \in (I_j \cup J) \cap
  R(p_j)$ and $K \in \intervals(R(p_j))$ be such that $b \in K$. Note that $K \cup
  I_j \cup J$ is an interval. If $b < a_j$, then, because $a_j \not\in R(p_j)$,
  either $g \notin K$ or $K$ has an upper bound if $J =
  [a_j,+\infty)$. Thus, $R'(p_j)$ is a progressing extension due to~\ref{itm:progr:b}
  or~\ref{itm:progr:c}. Finally, suppose that $b > a_j$.  By definition of $\rho$, there is a
  unique interval $I_j' \in \intervals(\TheBetterPost^{k-1}(R)(p_j))$ such
  that $I_j' \subseteq I_j$.  Moreover, $a_j \not \in I_j'$ and $a_j \ge
  \sup I_j'$. Thus, $K \cap I_j'$ is empty and $g \not \in I_j'$ or $I_j'$
  has an upper bound if $J = [a_j,+\infty)$.
  %By \autoref{proposition:succ} and
  %\autoref{lem:phi} we have $\overline{I_j'} \cap (P_{\mathcal{L}} \cup P_A) \neq
  %\emptyset$.
  Since $K \cup I_j \cup J \in \intervals$, $R'(p_j)$ is a
  progressing extension due to~\ref{itm:progr:b}
  or~\ref{itm:progr:c}.
\end{proof}
%
%We will prove that the indicator functions of sequences of
%$\subseteq$-increasing sets is stable except for a small number of cases.  To
%do so, we give a polynomial bound on the number subsequences are progressing
%extensions and prove that the indicator functions of nonprogressing extensions
%coincide.
%consider two cases
%based on subsequences which we call progressing extensions.

%\michael{I'm not sure what ``two cases'' is referring to.}
%\philip{could everything about progressing extensions be
%moved to subsection E? Lemma 22 and 23 are never invoked before that subsection,
%and it seems there is lots of notation here that could be postponed}

To conclude, we show the following bound.
\begin{lemma}\label{lemma:progress}
  Let $B_0, B_1, B_2,\dots \in \mathcal{C}$ be a sequence such that $B_i
  \subseteq B_{i+1}$ for all $i \in \N$. The set of $i \in \N$ such
  that $B_{i+1}$ is a progressing extension of $B_i$ has cardinality at most
  $|\mathcal{L}|^{\Oh(1)}$.
\end{lemma}

\begin{proof}
  Let $P \defeq P_{\mathcal{L}} \cup P_A$.
  %, \ie the codomain of the
  %indicator functions.
  First, observe that~\ref{itm:progr:c} can happen only if there exists some $\ell \in P$ such that 
  $\ell \in B_{i+1} \setminus B_i$, and thus at most $|P|$ times.

  Let $\phi_{B_i}(\intervals(B_i)) \subseteq P$ be the image of all
  intervals of $\intervals(B_i)$. Note that $\phi_{B_i} (\intervals(B_i)) \subseteq
  \phi_{B_{i+1}} (\intervals(B_{i+1}))$. A strict inclusion can happen at most
  $|P|$ times. Thus, we can assume that $\phi_{B_i} (\intervals(B_i)) =
  \phi_{B_i} (\intervals(B_{i+1}))$. Note that~\ref{itm:progr:a} can happen at
  most $|P|$ times due to \autoref{lemma:nothreesome} and because
  $\phi_{B_{i+1}}(I') \neq \emptyset$ for all $I' \in \intervals(B_{j+1})$.
  Indeed, for all $\ell \in P$, \autoref{lemma:nothreesome} tells
  us there are no pairwise distinct intervals $I_1, I_2, I_3 \in
  \intervals(B_{i+1})$ such that $\ell \in \phi_{B_{i+1}}(I_1)$, $\ell \in
  \phi_{B_{i+1}}(I_2)$, and $\ell \in \phi_{B_{i+1}}(I_3)$.
  
  Now, assume that~\ref{itm:progr:a} and~\ref{itm:progr:c} are not the case and that $\phi_{B_i}
  (\intervals(B_i)) = \phi_{B_i} (\intervals(B_{i+1}))$. Since~\ref{itm:progr:a} does not
  hold, we have $|\intervals(B_{i+1})| \le |\intervals(B_i)|$. Note that a
  strict inequality can happen at most $|P|$ times, so we can assume that
  $|\intervals(B_{i+1})| = |\intervals(B_i)|$. We define a function $f \colon
  \intervals(B_i) \to \intervals(B_{i+1})$. Recall that $B_i \subseteq
  B_{i+1}$. So, for every $I \in \intervals(B_i)$ there exists a unique $f(I) \in
  \intervals(B_{i+1})$ such that $I \subseteq f(I)$, and hence $\phi_{B_i}(I)
  \subseteq \phi_{B_{i+1}}(f(I))$. Uniqueness follows from maximality of
  intervals within $\intervals(B_{i+1})$. Since~\ref{itm:progr:a} does not hold, the function $f$ is a
  surjection. Moreover, $\ell \in \phi_{B_{i}}(I)$ iff $\ell \in
  \phi_{B_{i+1}}(f(I))$. Thus,~\ref{itm:progr:b} can happen at most $|P|$ times by
  \autoref{lemma:nothreesome}.
\end{proof}

As a corollary of how we computed $\reach_{p(a)}$ we get the following result
that will be needed in~\autoref{sec:parametric}.
\begin{corollary}\label{corollary:c-valid}
The mapping $\reach_{p(a)}$ is $\mathcal{C}$-valid.
\end{corollary}

\section{Parametric COCA reachability}
\label{sec:parametric}
\newcommand{\lep}{\underline{\mathit{ep}}}
\newcommand{\uep}{\overline{\mathit{ep}}}
\newcommand{\lclo}{\underline{\mathit{c}}}
\newcommand{\uclo}{\overline{\mathit{c}}}
\newcommand{\lbnd}{\underline{\mathit{bnd}}}
\newcommand{\ubnd}{\overline{\mathit{bnd}}}
\newcommand{\vars}[1]{( \lep_{#1}, \uep_{#1}, \lclo_{#1}, \uclo_{#1},%
                             \lbnd_{#1}, \ubnd_{#1} )}
\newcommand{\pred}[1]{\varphi_{#1}}
\newcommand{\succpred}{\varphi_{\mathrm{succ}}}
\newcommand{\ile}[2]{\mathrm{iflt}^{#1}_{#2}}

Finally, we consider the reachability problem $p(a) \steps{*} q(b)$
for parametric COCAs. We build on the proof of the last section for
guarded COCA. However, we ``guess'' under-approximations of reachability
functions by efficiently encoding them in the $\exists\forall$-fragment of
\foq instead of using iterative computation. We then invoke the following
result:
%\michael{Reachability sets or functions? Same for title below.}

\begin{theorem}[\cite{fr75,Son85}]\label{thm:foq}
  The language of sentences in the $\exists \forall$-fragment of \foq is decidable and in $\SIGTWO$.
\end{theorem}

In most of the section, we will work with a guarded~COCA $\W$ rather
than a parametric COCA $\Pp = (Q, T, \tau, X)$. One should think of
$\W$ as $\Pp^\mu$ for some valuation $\mu$. This valuation will be
guessed later on in our final \foq formula.

\subsection{Characterisation of the reachability function}

Recall the mapping $\reach_{p(a)} \in \mathcal{R}_Q$
from~\autoref{sec:ptime} defined by $\reach_{p(a)}(r) \defeq
\Post{p,r}{a}$. Note that $p(a) \steps{*} q(b)$ holds iff $b \in
\reach_{p(a)}(q)$.

Let $R \in \mathcal{R}_Q$ and recall the $\TheBetterPost$ function
from~\autoref{sec:ptime}. We say that mapping $R$ is a
\emph{reachability candidate} if $a \in R(p)$ and $\TheBetterPost(R)
\preceq R$. For example, the mapping $U$ defined by $U(r) \defeq \Q$,
for all $r \in Q$, is (trivially) a reachability candidate. As we will
show, $\reach_{p(a)}$ is also a reachability candidate.

The following provides a characterisation of $\reach_{p(a)}$ that will
allow us to encode it in \foq:

\begin{proposition}\label{proposition:characterisation}
  Let $R \in \mathcal{R}_Q$. It is the case that $\reach_{p(a)} = R$
  iff the following two conditions hold:
  \begin{itemize}
  \item $R$ is a reachability candidate, and

  \item $R \preceq R'$ for every reachability candidate $R' \in
    \mathcal{R}_Q$.
  \end{itemize}
\end{proposition}

\begin{proof}
  $\Rightarrow$) Let $R \defeq \reach_{p(a)}$. By definition, we have
  $a \in R$ and $\TheBetterPost(R) = R$. Hence, $R$ is a reachability
  candidate. Let $R' \in \mathcal{R}_Q$ be a reachability
  candidate. Let $q \in Q$ and $b \in R(q)$. We must show that $b \in
  R'(q)$. By definition of $R$, we have $p(a) \steps{\rho} q(b)$ for
  some run $\rho$. Note that $\TheBetterPost(R') \preceq R'$. Thus, by
  induction, we have $\TheBetterPost^{|\rho|}(R') \preceq R'$. Since
  $a \in R'(p)$, by \autoref{lem:post-succs} $b$ belongs to
  $\TheBetterPost^{|\rho|}(R')(q)$ and hence to $R'(q)$.
  
  $\Leftarrow$) Let $S_0 \in \mathcal{R}_Q$ be the mapping defined as
  $S_0(p) \defeq \set{a}$ and $S_0(r) \defeq \emptyset$ for every $r
  \neq p$. Let $S_i \defeq \TheBetterPost(S_{i-1})$ for every $i \geq
  1$. Since $R$ is a reachability candidate, we have $a \in R(p)$ and
  $\TheBetterPost(R) \preceq R$. In particular, we have $S_0 \preceq
  R$ and hence $S_1 = \TheBetterPost(S_0) \preceq \TheBetterPost(R)
  \preceq R$.

  By induction, we obtain $\bigcup_{i \in \N} S_i \preceq R$. By
  \autoref{lem:post-succs}, $\reach_{p(a)} = \bigcup_{i \in \N}
  \TheBetterPost(S_i)$. Thus, $\reach_{p(a)} \preceq R$. Since
  $\reach_{p(a)}$ is a reachability candidate, we have $R \preceq
  \reach_{p(a)}$ by assumption. This shows that $R = \reach_{p(a)}$.
\end{proof}

By \autoref{lemma:boundedunion} and \autoref{corollary:c-valid},
$\reach_{p(a)}(q)$ consists of at most $n$ intervals for every $q \in
Q$, where $n$ is polynomial in $|Q|$. We define $\mathcal{R}_Q^n
\subseteq \mathcal{R}_Q$ as the following set of mappings:
\begin{align*}
\mathcal{R}_Q^n \defeq \set{R : Q \to 2^{\Q} \mid |\intervals(R(q))| \le n \text{ for all } q \in Q}.
\end{align*}
In particular, we have $\reach_{p(a)} \in \mathcal{R}_Q^n$. Hence, we may rewrite \autoref{proposition:characterisation} as follows:

\begin{corollary}\label{corollary:characterisation}
Let $R \in \mathcal{R}_Q^n$. It is the case that $\reach_{p(a)} = R$ iff the
following two conditions hold:
\begin{itemize}
 \item $R$ is a reachability candidate, and

 \item $R \preceq R'$ for every reachability candidate $R' \in \mathcal{R}_Q^n$.
\end{itemize}
\end{corollary}

\subsection{Encoding intervals}

We now describe our logical encoding of intervals.
An interval $I$ is uniquely determined by its endpoints and their membership in $I$.
We represent
intervals as tuples $(b,t,\bot,\top)$, where $b,t \in \Q$ and
$\bot,\top \in \set{0,1,2}$. Such a tuple represents $\emptyset$ if
$(\bot \neq 2 \land \top \neq 2) \land (b > t \lor (b = t \land \neg (\top
= \bot = 1)))$ holds, and otherwise the unique interval
$I \in \intervals$ such that
\[
\bot = \begin{cases}
          0 & \text{if } \inf I = b \notin I, \\
          1 & \text{if } \inf I = b \in I, \\
          2 & \text{if } \inf I = -\infty,
         \end{cases}
\quad
\top = \begin{cases}
          0 & \text{if } \sup I = t \notin I, \\
          1 & \text{if } \sup I = t \in I, \\
          2 & \text{if } \sup I = +\infty.
         \end{cases}
\]
We write $\int(I)$ to denote the interval represented by the tuple $I =
(b,t,\bot,\top)$, \eg\ $\int(2,4,0,1) = (2,4]$, $\int(1,5,2,0) =
  (-\infty,5)$ and $\int(3,-2,1,0) = \emptyset$. Note that some
  intervals are not encoded uniquely, \eg\ $\int(0,5,2,0) =
  (-\infty,5)$ as well.

%% We further write $\inf I \defeq \inf \int(I)$ and $\sup I \defeq
%%   \sup \int(I)$.

Let $\encodings \defeq \Q^2 \times \set{0,1,2}^2$ be the set of all
interval encodings. We will represent each such encoding with four
variables from \foq, \ie\ $I_x = (x_b, x_t, x_\bot, x_\top)$. When
$I_x$ appears in a formula, we will assume that it is conjoined with
$\bigvee_{i=0}^2 (x_\bot = i) \wedge \bigvee_{i=0}^2 (x_\top =
i)$. Hence, variables $x_\top$ and $x_\bot$ evaluate only to
$\set{0,1,2}$. From now on, we can define \foq formulas that operate
on variables of the form $I_x$. As usual, outputs can be represented
with an extra variable, \eg\ $z = \min(x,y)$ % and $\max$
can be expressed as a
quantifier-free formula.
%\begin{align*}
%z = \min(x, y) &\defeq (x \le y) \rightarrow (z = x) \wedge (x > y) \rightarrow (z = y), \\
%z = \max(x, y) &\defeq (x \le y) \rightarrow (z = y) \wedge (x > y) \rightarrow (z = x).
%\end{align*}

Given two interval encodings $I = (b,t,\bot,\top) \in \encodings$ and
$I' = (b',t',\bot',\top') \in \encodings$, we can construct an encoding
$I'' \in \encodings$ such that $\int(I'') = \int(I) + \int(I')$:
\begin{align*}
I'' &\defeq (b+b',t + t',\minkowski(\bot,\bot'),\minkowski(\top,\top'));
\end{align*}
where $\minkowski \colon \set{0,1,2}^2 \to \set{0,1,2}$ is defined as:
\[
\minkowski(i,j) = \begin{cases}
                   2 & \text{if } \max(i,j) = 2; \\
                   \min(i,j) & \text{otherwise.}
                  \end{cases}
\]
Note that $I''$ can be written as a quantifier-free linear formula as
it involves addition and comparisons. By a slightly tedious case
distinction, it is further possible to construct a quantifier-free
linear formula that specifies an encoding $I'' \in \encodings$ such
that $\int(I'') = \int(I) \cap \int(I')$ (see \autoref{app:parametric}
of the appendix).

%% %
%% The following lemma shows that the operations behave as expected. The proof is a straightforward case analysis.
%% \begin{lemma}\label{lemma:intersection_sum}
%% Let $I, I' \in \encodings$. Then $\int(I \cap I') = \int(I) \cap \int(I')$ and $\int(I + I) = \int(I) + \int(I')$.
%% \end{lemma}

%% Similarly,
%%  \begin{align*}
%%    \ile{}{}(u,w,x,y,z) \coloneqq \; & (u < w) \rightarrow (z = x) \\
%%           \wedge \; & (u \geq w) \rightarrow (z = y)\\
%%         \minkowski(x,y,z) \coloneqq \; & (\max(x,y) = 2) \rightarrow (z = 2) \\
%%         \wedge \; & (\max(x,y) < 2) \rightarrow (\min(x,y,z)).
%% \end{align*}
%% We define the formulas for intersection and Minkowski sum. These are
%% quantifier-free formulas with $12$ variables but for the presentation we will
%% use vectors $I_x$.
%% \philip{reminder: fix intersection also here}
%% \guillermo{This will require reworking my fix above\dots}
%% \begin{align*}
%% \varphi_{\cap}(I_x,I_y,I_z) \coloneqq \; & \max(x_b, y_b,z_b) \wedge \min(x_t,y_t,z_t) \\
%% \wedge \; & \min(x_\bot,y_\bot,z_\bot) \wedge \min(x_\top,y_\top,z_\top); \\
%% \varphi_{+}(I_x,I_y,I_z) \coloneqq \; & (z_b = x_b+ y_b) \wedge (z_t =  x_t + y_t) \\
%% \wedge \; & \minkowski(x_\bot,y_\bot,z_\bot) \wedge \minkowski(x_\top,y_\top,z_\top).
%% \end{align*}

Formally, for every variable valuation $\mu$, we write $\mu(I_x)$ for
$(\mu(x_b), \mu(x_t), \mu(x_\bot), \mu(x_\top))$. We
can construct quantifier-free formulas $\varphi_+$ and $\varphi_\cap$
such that for every valuation $\mu$:
\begin{itemize}
\item $\mu(\varphi_{\cap}(I_x,I_y,I_z))$ holds iff $\mu(I_z) = \mu(I_x) \cap \mu(I_y)$,
  
\item $\mu(\varphi_{+}(I_x,I_y,I_z))$ holds iff $\mu(I_z) = \mu(I_x) + \mu(I_y)$.
\end{itemize}
For example, the latter is defined as:
\begin{multline*}
  \varphi_{+}(I_x,I_y,I_z) \defeq (z_b = x_b+ y_b) \land (z_t =
  x_t + y_t) \land {} \\ z_\bot = \minkowski(x_\bot,y_\bot) \land z_\top =
  \minkowski(x_\top,y_\top).
\end{multline*}
We will further use the following quantifier-free formulas for
membership, interval inclusion and emptiness checks:
\begin{align*}
\varphi_{\in}(x,I_y) \defeq \; & (y_b < x) \vee (y_b = x \wedge y_\bot = 1) \vee (y_\bot = 2) \\
\wedge \; & (x < y_t) \vee (y_t = x \wedge y_\top = 1) \vee (y_\top = 2); \\
\varphi_{\subseteq}(I_x,I_y) \defeq \; & (x_\bot = 2 \rightarrow y_\bot = 2) \wedge  (x_\top = 2 \rightarrow y_\top = 2) \\
\wedge \; & (y_b < x_b) \vee (y_b = x_b \wedge x_\bot \le y_\bot) \vee (y_\bot = 2) \\
\wedge \; & (x_t < y_t) \vee (y_t = x_t \wedge x_\top \le y_\top) \vee (y_\top = 2);\\
\varphi_{\emptyset}(I_x) \defeq \; & (x_\bot \neq 2) \wedge (x_\top \neq 2) \\
\wedge \; & \left((x_b > x_t) \vee (x_b = x_t \wedge \neg(x_\top = x_\bot = 1))\right).
\end{align*}

For every valuation $\mu$, the following holds by definition:
\begin{itemize}
\item $\mu(\varphi_{\in}(x,I_y))$ holds iff $\mu(x) \in \int(\mu(I_y))$;

\item $\mu(\varphi_{\subseteq}(I_x,I_y))$ holds iff $\int(\mu(I_x)) \subseteq \int(\mu(I_y))$;
   
\item $\mu(\varphi_{\emptyset}(I_x))$ holds iff $\int(\mu(I_x)) = \emptyset$.
\end{itemize}

\subsection{Unions of intervals}

We will consider a union of intervals to represent $\reach_{p(a)}$. We will represent them as vectors of variables. We write $\bI_x = (I_{x_1},\ldots,I_{x_n})$ to denote a vector of $4n$ variables, where $n$ is the bound from \autoref{corollary:characterisation}. Given a valuation $\mu$, we write $\mu(\bI_x)$ for the valuation of all variables in $\bI_x$, and we write $\int(\mu(\bI_x)) \defeq \bigcup_{i = 1}^n \int(\mu(I_{x_i}))$.

We define formulas dealing with union of intervals $\bI_x$ and generalising the formulas on $I_x$. For membership, we define:
\begin{align*}
\psi_{\in}(x, \bI_y) \defeq \bigvee_{i = 1}^n \varphi_{\in}(x,\bI_{y_i}).
\end{align*}
For Minkowski sums, we will only need to sum a union of intervals with
a single interval:
\begin{align*}
\psi_{+}(\bI_x, I_y,\bI_z) \defeq \bigwedge_{i = 1}^n \varphi_{+}(I_{x_i},I_y,I_{z_i}).
\end{align*}
Similarly, we will intersect a union of intervals with a single interval. Using $I \cap \bigcup_{a \in A} I_a = \bigcup_{a \in A} I \cap I_a$, 
% We write $\bI_{x,y}$ for the $4n^2$ variables. These are $n^2$ intervals that we denote $I_{x_i,y_j}$ for all $1 \le i,j \le n$.
we can define:
\begin{align*}
\psi_{\cap}(\bI_x,I_y, \bI_{z}) \defeq \bigwedge_{i = 1}^n \varphi_{\cap}(I_{x_i}, I_{y}, I_{z_i}).
\end{align*}
The formula for interval inclusion is defined as:
\begin{align*}
\psi_{\subseteq}(\bI_x, \bI_y) \defeq \bigwedge_{i = 1}^n \bigvee_{j = 1}^n \varphi_{\subseteq}(I_{x_i},I_{y_j}).
\end{align*}
Note that this formula does not exactly express inclusion. For example, let $n \defeq 2$. Suppose $\bI_x$ encodes $(1,5) \cup (4,9)$ and $\bI_y$ encodes $[1,4) \cup [4,10]$. We have $(1,5) \cup (4,9) \subseteq [1,4) \cup [4,10]$ even though $(1,5)$ is not a subset of any of the latter intervals.

Therefore, we need a formula to express that $\bI_x$ is decomposed into maximal intervals.
Let us start with a formula to express that an interval does not ``prolong'' another one (\eg\ as opposed to $[4,10]$ which prolongs $[1,4)$):
\begin{align*}
\varphi_{dis}(I_x,I_y) \defeq & (x_t = y_b) \rightarrow (x_\top = y_\bot = 0).
\end{align*}
Notice that this formula alone does not express that $I_x$ and $I_y$ are disjoint.
% Here we write $I_{y_{i,j}}$ for auxiliary intervals
The following formula expresses that $\bI_x$ consists of maximal intervals, using some auxiliary intervals $I_{y_{i,j}}$:
\begin{multline*}
  \psi_{max}(\bI_x) \defeq \bigwedge_{1 \le i, j \le n} \left( \varphi_{\cap}(I_{x_i},I_{x_j},I_{y_{i,j}}) \rightarrow \varphi_{\emptyset}(I_{y_{i,j}}) \right) \\[-8pt]
  \wedge \varphi_{dis} (I_{x_i},I_{x_j}).
\end{multline*}
The above formula has a quadratic number of variables, but we will omit variables $I_{y_{i,j}}$ in $\psi_{max}(\bI_x)$ for readability.
For every valuation $\mu$, it is readily seen that $\mu(\psi_{max}(\bI_x))$ holds iff $\int(\mu(I_{x_1})), \ldots, \int(\mu(I_{x_n}))$ are the maximal intervals from $\intervals(\int(\mu(\bI_x)))$, plus possibly some empty intervals. Thus, we say that $\mu(\psi_{max}(\bI_x))$ is decomposed into maximal intervals.

More formally, the following holds for every valuation $\mu$:
\begin{itemize}
\item $\mu(\psi_{\in}(x,\bI_y))$ holds iff $\mu(x) \in \int(\mu(\bI_y))$;

\item $\mu(\psi_{+}(\bI_x,I_y,\bI_z))$ holds iff $\int(\mu(\bI_x)) + \int(\mu(I_y)) = \int(\mu(\bI_z))$;

\item $\mu(\psi_{\cap}(\bI_x,I_y, \bI_{z}))$ holds iff 
  \(
 \int(\mu(\bI_x)) \cap \int(\mu(I_y)) = \bigcup_{i = 1}^n \int(I_{z_i});
 \)
 
\item $\mu(\psi_{\subseteq}(\bI_x,\bI_y) \cap \psi_{max}(\bI_y))$ holds iff $\bI_y$ is decomposed into maximal intervals and
 $
 \int(\bI_x) \subseteq \int(\bI_y).
 $
\end{itemize}

\subsection{Encoding the reachability function}
\label{subsection:encoding}

For every $z \in \Z \cup X$, let $I_z^+, I_z^-, I_z^0 \in \encodings$ be
defined as $I_z^+ \defeq (0,z,0,1)$, $I_z^- \defeq (z,0,1,0)$, and
$I_z^0 = (0,0,1,1)$. These encode $(0, z]$, $[z, 0)$ and $[0,
    0]$ which are the possible scalings of $z$ depending on its sign. For every $q \in Q$, let $I_q \in \encodings$ be such that $I_q \defeq (q_b,q_t,q_\bot,q_\top)$ encodes $\tau(q)$. For example, if $\tau(q) = (x, 10]$, then we can take $I_q \defeq (x, 10, 0, 1)$. Note that $q_b$ and $q_t$ can be variables corresponding to parameters.

Given $\bI_x$, for all transition $t \in T$, we define $\bI_t$ through:
\begin{align*}
\psi_{t}(\bI_x,\bI_t) \defeq \; & (\Effect{t} > 0) \rightarrow \psi_{+}(\bI_{x},I_{\Effect{t}}^+,\bI_y)\\
\wedge \; & (\Effect{t} < 0) \rightarrow \psi_{+}(\bI_{x},I_{\Effect{t}}^-,\bI_y) \\
\wedge \; & (\Effect{t} = 0) \rightarrow \psi_{+}(\bI_{x},I_{\Effect{t}}^0,\bI_y) \\
\wedge \; & \psi_{\cap}(\bI_y,I_{\Out{t}},\bI_t).
\end{align*}
Note that the above formula uses auxiliary variables $\bI_y$, $I_{\Effect{t}}^+$, $I_{\Effect{t}}^-$, $I_{\Effect{t}}^0$, and $I_{\Out{t}}$.

Let $\bI_Q = (\bI_{r})_{r \in Q}$ be the vector of $|Q| \cdot 4n$
variables, where each $\bI_{r}$ consists of $n$ intervals. Similarly, we define $\bI_T$ as a vector of $|T| \cdot 4n$ variables. We define the formula
\begin{align*}
\psi_{succ}(\bI_Q,\bI_{T}) \defeq \bigwedge_{t \in T} \psi_{t}(\bI_{\In{t}},\bI_{t}).
\end{align*}

Given a valuation $\mu$ and a vector $\bI_Q$, we define the mapping
$\mu(\bI_Q)$ as $q \mapsto \int(\mu(\bI_{q}))$. The
following holds by definition:

\begin{lemma}\label{lem:encoding-succs}
For every valuation $\mu$, $\psi_{succ}(\bI_Q,\bI_{T})$ holds iff for every $r \in Q$ the following holds:
\begin{align*}
\TheBetterPost_r(\mu(\bI_Q)) = \mu(\bI_Q)(r)\ \cup\ \bigcup_{\mathclap{\substack{t \in T,\\ \Out{t} = r}}} \mu(\bI_t).
\end{align*}
\end{lemma}

This allows us to define the following formula:
\begin{align*}
\psi_{cand}(\bI_Q) \defeq \; & \psi_{succ}(\bI_Q,\bI_{T}) \\
\wedge \; & \bigwedge_{r \in Q} \;\; \bigwedge_{\mathclap{\substack{t \in T \\ \Out{t} = r}}} \psi_{\subseteq}(\bI_{t},\bI_{r}) \cap \psi_{max}(\bI_{r})
\wedge \psi_{\in}(a,\bI_{p}).
\end{align*}
Again, by definition we obtain:
\begin{lemma}
For all valuation $\mu$, it is the case that $\psi_{cand}(\bI_Q)$ holds iff $\mu(\bI_Q)$ is a reachability candidate.
\end{lemma}

Let $Q'$ be a disjoint copy of $Q$.
We write $\bX$ for the vector of variables corresponding to parameters of the COCA,
$\bY$ for the vector of all remaining variables used in the quantifier-free formula $\psi_{cand}(\bI_Q)$ (not just $\bI_Q$), and $\bZ$ for the vector of remaining variables
used in the quantifier-free formula $\psi_{cand}(\bI_{Q'})$.
Further, recall $I_p$ encodes $\guard{p}$. We define:
\begin{align*}
\psi \defeq \exists \bX\ \exists \bY\ \forall \bZ \;  &\psi_{cand}(\bI_Q) \wedge \psi_{cand}(\bI_{Q'}) \wedge \psi_{\in}(a, I_p) \\
\wedge \;  \bigwedge_{r \in Q} &\psi_{\subseteq}(\bI_{r}, \bI_{r'}) \wedge \psi_{max}(\bI_{r}) \wedge \psi_{\in}(b, \bI_q).
\end{align*}

Since $\psi$ holds iff $p(a) \steps{*} q(b)$, by \autoref{corollary:characterisation} we
obtain:

\begin{theorem}
  The (existential) reachability problem for parametric COCAs belongs
  in $\SIGTWO$.
\end{theorem}

\subsection{Integer valuations}

We briefly consider parametric COCAs where only updates can be parameterised.
In this setting, a rational valuation that witnesses reachability can be turned into an \emph{integer} valuation witnessing reachability. This follows by rescaling
the factors of the witnessing run so that it remains admissible.

\begin{lemma}\label{lemma:multiply-solutions}
  Let $\mu$ be a valuation under which $p(a) \steps{*} q(b)$. For any
  valuation $\mu'$ such that $\mu'(x) = \lambda \mu(x)$ with
  $\lambda \in \N_{\geq 1}$, it is the case that $p(a) \steps{*} q(b)$
  under $\mu'$.
\end{lemma}

\begin{proof}
  Let $\lambda \in \N_{\geq 1}$ and let $\mu'$ be defined
  w.r.t.\ $\mu$ and $\lambda$. Let $s(v) \steps{\alpha t} s'(v')$ be
  consecutive configurations from the run $\rho$ witnessing $p(a)
  \rightarrow_\rho q(b)$ under valuation $\mu$. If $\Effect{t} \in
  \Z$, then no rescaling is needed as the update of $t$ is
  nonparametric. Otherwise, we have $v' - v = \alpha \cdot
  \mu(\Effect{t})$. Let $\beta \defeq \alpha / \lambda$. Since $\alpha
  \in (0, 1]$ and $\lambda \geq 1$, we have $\beta \in (0,
  1]$. Therefore, we have $s(v) \steps{\beta t} s'(v')$. Hence, by
rescaling each transition, we obtain a run $\rho'$
such that $p(a) \steps{\rho'} q(b)$ under $\mu'$.
\end{proof}

Now, consider a rational valuation $\mu$ witnessing $p(a) \steps{*}
q(b)$. Since $\mu$ is rational, each parameter value $\mu(x)$ can be
represented as a fraction $a_x / b_x$. By
\autoref{lemma:multiply-solutions}, we know that valuation $\mu'(x) =
\lambda \mu(x)$, where $\lambda \defeq \prod_{x\in X} b_x$, also
witnesses reachability. Moreover, it is integral, hence:

\begin{corollary}
  The (existential) reachability problem for parametric COCAs, where
  valuations must be integral, is equivalent to the rational variant
  if guards are nonparametric.
\end{corollary}

\subsection{Hardness result and acyclic parametric COCA}
\label{sec:dag-parametric}

To conclude our treatment of parametric COCAs, we establish
\NP-hardness of the reachability problem, and a matching \NP\ upper
bound for the special case of acyclic COCAs.

\begin{theorem}\label{thm:acyclic-np-complete}
  The reachability problem for parametric COCAs is \NP-hard. Moreover,
  this holds if the underlying graph is acyclic, and if parameters
  occur on transitions, guards or both. In the latter special cases,
  the problem is \NP-complete.
\end{theorem}

\begin{proof}
  We first show membership in \NP\ for the acyclic case. To determine whether
  $p(a) \rightarrow_{*} q(b)$ under some valuation $\mu$, we encode the
  problem into an existential \foq sentence. Since the language of such
  sentences is known to be in \NP{}~\cite{Son85} the result will follow.
  
  We start with the observation that $\Post{\pi}{a}$ is an interval for every
  path~$\pi$ since it can be described
  %. Indeed, this is because $\Post{\pi}{a}$ is computed by applying 
  as a chain of Minkowski sums and intersections
  %MIUN-operations
  %: Minkowski sum and intersection 
  starting from $[a,a]$.
  Recall the formula $\psi_t(\bI_x,\bI_t)$ from \autoref{subsection:encoding}.
  %To guess the transitions 
  We assume %that 
  there is a bijection $f$ between $T$ and $\set{1,\ldots,|T|}$ and
  %we 
  identify $t\in T$ with $f(t)$.
  Let
  %We define the formula 
  $\varphi_{f(t)}(I_x,I_t)$ be defined in the same way as $\psi_t(\bI_x,\bI_t)$
  but replacing vectors of intervals %are replaced 
  with intervals and $\psi_+$, $\psi_\cap$
  %are replaced 
  with $\varphi_+$, $\varphi_\cap$. Intuitively, $\varphi_{f(t)}(I_x,I_t)$ means that
  $I_t = \Post{t}{I_x}$.
  
  Given two transitions $t,t' \in T$ we also need a formula 
  that checks $\Out{t} = \In{t'}$. We define this by
  \(
   \varphi_{next}(x,y) \coloneqq \; \bigwedge_{i,j = 1}^{|T|} (x = i \wedge y = j) \rightarrow t_{ij},
  \)
where $t_{ij}$ is a Boolean constant that is true iff $\Out{f^{-1}(i)} = \In{f^{-1}(j)}$.
Similarly, we define $\varphi_{init}(x)$, $\varphi_{fin}(x)$ as the formulas which are true iff
$x$ is a transition that starts in state $p$ and finishes in state $q$, respectively.
  
%   Let $S(p) \defeq
%   \set{a}$ and $S(r) \defeq \emptyset$ for every $r \neq p$. From \autoref{lem:post-succs}, we have $c \in
%   \TheBetterPost^{|Q|}(S)(q)$ iff $p(a) \rightarrow_{\rho} q(c)$ for some run
%   $\rho$ such that $|\rho| \leq |Q|$. Since the parametric COCA is acyclic, the
%   latter holds iff there is any admissible run (without restriction on the
%   length). Using \autoref{lem:encoding-succs}, we can write $|Q|$ times the
%   sentence encoding the decomposition
%   of $\TheBetterPost$. In the end we assert that $b$ is in the final set.
%   One issue we have to resolve is that in the formula $\psi_{succ}(\bI_Q,\bI_T)$
%   the variables $\bI_T$ do not encode a mapping like $\bI_Q$. Therefore, we will
%   guess some intervals in $\bI_T$ to encode an under-approximation of $\TheBetterPost$.
%   Technically, after every ``application'' of $\TheBetterPost$, we will
%   additionally guess an interval per state from its output and continue from
%   there onward.

%   We write $\bJ_1,\dots,\bJ_{|Q|}$ to denote $|Q|$ disjoint copies of $\bI_T$
%   and $\bI_0, \ldots, \bI_{|Q|}$ to denote $|Q|+1$ disjoint copies of $\bI_Q$.
  
%   For $k \in \{1,\dots,|Q|\}$ and $\mathbf{V} = (v_1,v_2, \dots)$ a vector
%   $\{1,\dots,n\}^{|Q|}$ we write $\psi_{sel}(\bJ_k,\bI^{\mathbf{V}})$ for
%   the quantifier-free formula which for all $r_i \in Q$ asserts that
%   $\bI_{r_i}$ encodes the $v_i$-th interval from $\bJ_k$ and the rest empty
%   intervals.
%   Further, we write $\bI_0$ to denote the encoding of $S$ as
%   defined above.

Let $\bT = (t_1,\ldots,t_{|Q|})$ and $\bI = (I_1\ldots,I_{|Q|})$ be vectors of variables representing $|Q|$ transitions and $|Q|$ intervals.
We write $I_0$ for the vector $(a,a,1,1)$ encoding $[a,a]$.
% Without loss of generality, we assume that $f(p) = 1$ and $f(q) = 2$.
Finally, let $\bX$ be the vector of variables encoding the parameters of the COCA. We define the final formula
  \begin{align*}
   \exists \bX \ \exists \bT \ \exists \bI \ & \varphi_{init}(t_1) 
   %\; \wedge
   %\; \varphi_{fin}(t_{|Q|}) \\
   %
   \wedge \bigwedge_{i = 1}^{|Q|} \bigvee_{j = 1}^{|T|} t_i = j \wedge
   \varphi_{j}(I_{i-1},I_i)\\
   %
   %\; \bigwedge_{i = 2}^{|Q|} \varphi_{next}(t_{i-1},t_i)\\
   %
   %&\wedge \; \bigwedge_{i = 1}^{|Q|} \bigvee_{j = 1}^{|T|} (t_i = j)
   %\rightarrow \varphi_{t_i}(I_{i-1},I_i) \\
   %
   {} \wedge {} & \bigvee_{m = 1}^{|Q|} \varphi_{fin}(t_m) \wedge \varphi_{\in}(b,I_m)
   \wedge \bigwedge_{i = 1}^{m-1} \varphi_{next}(t_{i},t_{i+1}).
    %\wedge \;
    %\varphi_{\in}(b,I_{|Q|}).
     %
  \end{align*}
  %Since all runs in an acyclic parametric COCA have length at most $|Q|$,
  By definition,
%It follows that 
  this sentence %formula is satisfiable 
  holds iff $p(a) \rightarrow_{*} q(b)$. Indeed, since the COCA is
  acyclic all runs are of length at most $|Q|$.
  Thus, in the first line
  we ``guess'' transitions per state and ``compute'' their reachability
  set; then, in the second line we check whether they form a path $\pi$ and
  whether $p(a) \rightarrow_\pi q(b)$.
  %\michael{The above formula seems outdated (and
  %an \texttt{\textbackslash autoref} is).}
%   \guillermo{The new encoding needs double checking}
  
  It remains to show \NP-hardness for the acyclic case,
  where parameters either only occur on transitions or  on
  guards.
  In both cases, we give a reduction from 3-SAT.
  Let $\varphi = \bigwedge_{1 \leq j
  \leq m} C_j$ be a 3-CNF formula over variables $X = \{x_1, \ldots,
  x_n\}$. 
  
  Let us give two acyclic parametric COCAs $\Pp$ and $\Pp'$, both with parameters $X$.
  Each one will guess an assignment to $X$, and check whether it satisfies
  $\varphi$.
  Additionally, $\Pp$ uses parameters only on guards; $\Pp'$, only on updates.
  We only sketch $\Pp$. While
  the exact gadgets used in the construction for $\Pp'$ differ,
  they perform the same functions (see the appendix).

  \begin{figure}[!h]
    \begin{center}
      \begin{tikzpicture}[auto, thick, transform shape, scale=0.9]
  \tikzstyle{astate} = [state, minimum size=15pt, inner sep=0pt];
  %% Guessing
  \node[astate, label={above:$[0, 0]$}]                                         (g0) {$p_i$};
  \node[astate, label={below:$[1, 1]$}, below right=0.75cm and 0.25 of g0]  (g1) {};
  \node[astate, label={below:$[x_i, x_i]$},     right=0.75cm of g1]                 (g2) {};
  \node[astate, label={above:$[x_i, x_i]$}, right=0.75cm of g0]                 (g3) {};
  \node[astate, label={above:$[0, 0]$},     right=0.75cm of g3]                 (g4) {$q_i$};

  \path[->]
  (g0) edge node[swap] {$1$}    (g1)
  (g0) edge node       {$0$}    (g3)
  (g3) edge node       {$0$}    (g4)
  (g1) edge node       {$0$}    (g2)
  (g2) edge node[swap] {$-1$}   (g4)
  ;

  %% Checking
  \node[astate, label={above:$[0, 0]$}, right=1.5cm  of g4, yshift=-0.5cm] (c0) {$r_j$};
  \node[astate, label={above:$[x_2, x_2]$}, right=1cm    of c0] (c2) {};
  \node[astate, label={above:$[x_1, x_1]$}, above=0.75cm of c2] (c1) {};
  \node[astate, label={above:$[x_4, x_4]$}, below=0.75cm of c2] (c3) {};
  \node[astate, label={above:$[0, 0]$}, right=1cm    of c2] (c4) {$s_j$};

  \path[->]
  (c0) edge node       {$1$} (c1)
  (c0) edge node       {$1$} (c2)
  (c0) edge node[swap] {$0$} (c3)
  (c1) edge node[above left, xshift=10pt, yshift=+5pt] {$-1$}  (c4)
  (c2) edge node       {$-1$}  (c4)
  (c3) edge node[below left, xshift=10pt, yshift=-5pt] {$0$}  (c4)
  ;
\end{tikzpicture}
    \end{center}
    \caption{Gadgets of the reduction from 3-SAT, where parameters occur only on guards, for: variable $x_i$ (\emph{left}) and clause $C_j = (x_1 \lor x_2 \lor \neg x_4)$ (\emph{right}).}
    \label{fig:np:complete-guards}
  \end{figure}
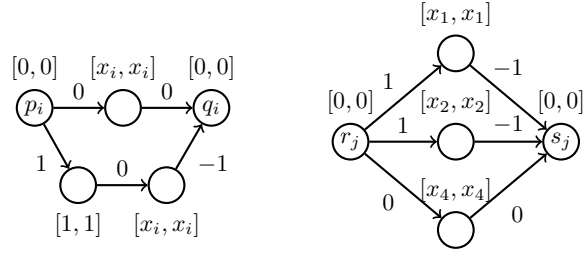

  The first part is done by sequentially composing $n$ copies of the left gadget depicted
  in \autoref{fig:np:complete-guards}. The gadget functions as follows:
  (1)~state $p_i$ is entered with counter value $0$; (2)~the counter is set to $x_i$;
  (3)~membership of the counter value in $\{0,1\}$ is checked; and (4)~the counter
  is reset to zero upon leaving to $q_i$.
  The only way to traverse the chain of
  $n$ such gadgets from $p_1$ to $q_n$ is to have $x_i \in \{0, 1\}$
  for each $x_i \in X$.

  The second part is achieved by chaining a gadget for each clause
  similar to the one depicted on the right-hand side of
  \autoref{fig:np:complete-guards} for $C_j = (x_1 \lor x_2 \lor \neg
  x_4)$. 
  In words, it (1)~enters state $r_j$ with the counter value set to $0$;
  (2)~nondeterministically picks a variable $x_i$ of some
  literal of $C_j$ and increments the counter by $x_i$; (3)~checks
  whether the counter holds the right value w.r.t.\ the literal
  polarity; and (4)~resets the counter to zero upon leaving to state
  $s_j$. Thus, the chain of gadgets can be traversed from $r_1$ to
  $s_m$ iff $\varphi$ is satisfied by the assignment.

  Altogether, these statements are equivalent:
  (1)~formula $\varphi$ is satisfiable; (2)~there exists a
  valuation $\mu \colon X \to \Q$ such that $p_1(0) \steps{*} s_m(0)$
  holds in $\Pp^\mu$; and (3)~there exists a
  valuation $\mu' \colon X \to \Q$ such that $p_1(0) \steps{*} s_m(0)$
  holds in $\Pp'^{\mu'}$
\end{proof}

\section{Conclusion}
\label{sec:conclusion}
In this work, we have introduced COCA and guarded COCA as
over-approximations of SOCA, and we have given efficient algorithms
for their reachability problems. For both models, the only lower bound
we are aware of is the $\NL$-hardness that follows trivially from the
directed-graph reachability problem. It thus remains open whether our
algorithms are computationally optimal. Additionally, we have shown
that the reachability problem lies in the polynomial hierarchy for
parametric COCA, in contrast to the discrete variant whose
decidability is unknown. We leave open whether the reachability
problem for parametric COCA with \emph{integer}-valued parameters is
decidable in general (\ie, when guards can also be parametric).

%\section*{Acknowledgment}
%The authors would like to thank...no one.

%% Bibliography
\IEEEtriggeratref{13} %% Balances the last two-columns as suggested in
                      %% the guidelines
\bibliographystyle{IEEEtran}
\bibliography{IEEEabrv,bibliography}

%% Appendix
\clearpage
%\onecolumn
\appendix
\label{sec:appendix}
%%%%%%%%%%%%%%%%%%%
%% APPENDIX OF SECTION III %%
%%%%%%%%%%%%%%%%%%%

\subsection{Missing proofs of \autoref{sec:nc2}}

\lemmaGraphReach*

\begin{proof}
  Let us explain how one can check emptiness of $S$
  when only one of the conditions is required by reducing
  the problem to standard graph reachability, which is in $\NL \subseteq \NC^2$.

  For each condition, we focus only on one of the two
  stated cases.
  The other case will follow similarly.
  Each time, we will give a graph $H$ such that
  reachability between two fixed nodes of $H$ corresponds
  to reachability from $p$ to $q$ in $G$ via a
  path that satisfies the condition.

  \medskip
  \noindent\ref{itm:has-pos} Let us define $H$ as joining
  $G$ with a modified copy $\overline{G}$ 
  which remains in $\overline{G}$ with nonpositive edges,
  and moves into $G$ with positive edges.
  More formally, let $H \defeq (Q',E')$, where
  $Q' \defeq Q \cup \{\overline{q} \mid q \in Q\}$.
  Further, $E' \defeq E \cup \{(\overline{p},z,\overline{q}) \mid 
  (p,z,q) \in E, z \leq 0\} \cup \{(\overline{p},z,q) \mid 
  (p,z,q) \in E, z > 0\}$.
  It is easy to see any path $\overline{\pi}$ from $\overline{p}$ to $q$ in $H$
  corresponds to a path $\pi$ from $p$ to $q$ in $G$ such that
  $\Effectp{\pi} \neq 0$.

  \medskip
  \noindent\ref{itm:no-neg} We set $H \defeq (Q, E')$,
  where $E' \defeq \{(p',z,q') \in E \mid z \leq 0\}$.
  Clearly, reachability from $p$ to $q$ in $H$ is
  equivalent to reachability from $p$ to $q$ in $G$ via
  a path $\pi$ with $\Effectp{\pi} = 0$.

  \medskip
  \noindent\ref{itm:first:pos} We again 
  define $H$ as joining $G$ with a modified copy $\overline{G}$.
  In $\overline{G}$, we omit all positive edges,
  while edges with weight zero remain in $\overline{G}$
  and negative edges lead to $G$. 
  Formally, let $H \defeq (Q', E')$,
  where $Q' \defeq Q \cup \{\overline{q} \mid q \in Q\}$ and 
  $E' \defeq E \cup \{(\overline{p},0,\overline{q}) \mid (p,0,q) \in E\}
  \cup \{(\overline{p},z,q) \mid (p,z,q) \in E \wedge z < 0\}$.
  A path $\overline{\pi}$ from $\overline{p}$ to $q$
  in $H$ corresponds to a path $\pi$ from $p$ to $q$ in $G$
  such that $\Effect{\first{\pi}} < 0$.

  \medskip
  \noindent\ref{itm:last:neg} We again join $G$
  with a modified copy $\overline{G}$.
  Now, $\overline{G}$ omits all positive edges,
  and for each negative edge in $G$,
  we add a copy that leads from $G$ to $\overline{G}$.  
  We define $H \defeq (Q', E')$
  with $Q' \defeq Q \cup \{\overline{q} \mid q \in Q\}$ and
  $E' \defeq E \cup \{(p,z,\overline{q}) \mid (p,z,q) \in E \wedge z < 0\}
  \cup \{(\overline{p},z,\overline{q}) \mid (p,z,q) \in E \wedge z \leq 0\}$.
  A path $\overline{\pi}$ from $p$ to $\overline{q}$
  in $H$ corresponds to a path $\pi$ from $p$ to $q$ in $G$
  such that $\Effect{\last{\pi}} < 0$.

  \medskip

  When we wish to require several conditions at once,
  note that for each condition, we constructed a graph $H$ from a given input graph $G$.
  To require many conditions at once, we simply apply the transformations
  for each condition sequentially, and obtain a graph $H'$
  such that paths of $H'$ satisfy all
  imposed conditions and correspond to paths in the original graph $G$. Observe that $H'$ is of polynomial size.

  Finally, let us argue that the following value can be computed in $\NC^2$: $\mathrm{opt}\{w(\pi) \mid \pi \in S \text{ and
  } \abs{\pi} \leq \abs{Q})\}$,
  where $\mathrm{opt} \in \{\min,
  \max\}$ and $w \in \{\Delta^+, \Delta^-\}$. Let us first deal with $w = \Delta^+$, for which it suffices to
  treat edges with negative weight as having zero weight.

  The
  problem of finding a shortest weighted path in a graph with
  edges of nonnegative weights is in $\NC^2$ (\eg, see~\cite[Example 12.4]{bovet1994introduction}).
  The procedure relies on the fact that 
  there must be an acyclic shortest path, and hence that it suffices to consider paths of length at most $\abs{Q}$.
  We can easily adapt the standard procedure for maximisation.
  Indeed, it successively minimises paths of length $1, 2, 4, 8, \dots, |Q|$.
  By maximising rather than minimising, it follows
  that the claim holds for $\mathrm{opt} \in \{\min,
  \max\}$. Note that finding a longest simple path is \NP-complete, while we obtain $\NC^2$ because we may find a nonsimple path (of length at most $|Q|$).
  
  Since each of the conditions \ref{itm:has-pos}--\ref{itm:last:neg}
  is achieved by transforming the input graph into another graph of polynomial size,
  this holds also if we require any subset of these conditions.

   The case of $\Delta^-$ can be handled similarly, \eg\ by flipping the sign of the weights and the optimisation type ($\max$/$\min$).
\end{proof}

\propScaleDown*

\begin{proof}
  We will show that for all $i \in \{0, \ldots, |\rho|\}$, either $a
  \leq a + \effect{\FromTo{\beta \rho}{1}{i}} \leq
  a+\effect{\FromTo{\rho}{1}{i}}$ or $a \geq a + \effect{\FromTo{\beta
      \rho}{1}{i}} \geq a+\effect{\FromTo{\rho}{1}{i}}$.  Since $a +
  \effect{\FromTo{\rho}{1}{i}} \in \tau$ holds by the admissibility of
  $\rho$ from $a$, it follows that $a + \effect{\FromTo{\beta
      \rho}{1}{i}} \in \tau$, and so that $\beta \rho$ is admissible.

  By definition, we have $a + \effect{\FromTo{\beta \rho}{1}{i}} = a +
  \beta \effect{\FromTo{\rho}{1}{i}}$.  Additionally, $\beta \in
  \ZeroOne$. Hence, if $\effect{\FromTo{\rho}{1}{i}} \geq 0$, then we
  have $a + \effect{\FromTo{\rho}{1}{i}} \geq a +
  \effect{\FromTo{\beta \rho}{1}{i}} \geq a$.  If
  $\effect{\FromTo{\rho}{1}{i}} < 0$, then $a +
  \effect{\FromTo{\rho}{1}{i}} \leq a + \effect{\FromTo{\beta
      \rho}{1}{i}} < a$, so we are done.
\end{proof}

\lemmaIntervals*

\begin{proof}
  We only prove $(a, b) \subseteq \Post{p,q}{a}$ as the other
  inclusion is symmetric. We assume that $a < b$ as we are otherwise
  done. Let $c \in (a, b)$. Since $b \in \cl{\Post{p,q}{a}}$, there
  exists $b' \in \Post{p,q}{a}$ such that $b' \in [c, b]$. Let $\rho$
  be an admissible run from $p(a)$ to $q(b')$. By definition,
  $\Effect{\rho} = b' - a$. Let $\beta \defeq (c - a) / (b' - a) \in
  (0, 1]$. By \autoref{claim:scaling-down-gives-run}, $\beta \rho$ is
    admissible from $p(a)$. Since $\Effect{\beta \rho} = c - a$, this
    concludes the proof of the main claim.
\end{proof}

\propObs*

\begin{proof}
  \leavevmode
  \begin{enumerate}[label=(\alph*)]
  \item Let $a \defeq \inf \tau$ and $c \defeq
    \sup\cl{\Post{p,q}{}}$. By \autoref{lemma:intervals}, we have $(a,
    c) \subseteq \cl{\Post{p,q}{a}}$. Since the latter is closed by
    definition, we have $\inf \cl{\Post{p,q}{a}} = \inf \tau$.

  \item The proof is symmetric to~(a).
    
  \item Since $v \notin \{\inf \tau, \sup \tau\}$, there is a small
    enough $\epsilon \in (0, 1]$ such that $v + |\epsilon \cdot
      \Effect{\FromTo{\rho}{1}{i}}| \in \tau$ for all $i \in \set{1,
        \ldots, |\rho|}$. By definition, $\epsilon \rho$ is admissible
      from $v$. Let $v_\beta \defeq v + \beta \cdot \Effect{\rho}$. By
      \autoref{claim:scaling-down-gives-run}, $v \steps{\beta \rho}
      v_\beta$ is admissible for every $\beta \in (0,
      \epsilon]$. Moreover, $\lim_{\beta \to 0} v_{\beta} =
        \lim_{\beta \to 0} v + \beta \cdot \Effect{\rho} = v$.\qedhere
  \end{enumerate}
\end{proof}

\lemmaSupfInf*

\begin{proof}
  Let $\theta = t \pi$ where $t$ is the first transition of $\theta$
  and $\pi$ is the remaining path. Let $r \defeq \In{t}$. Since $a \in
  \enab{\Paths{p}{r}}$, there is an admissible run $\rho_1$ from
  $p(a)$ that ends in state $r$. Similarly, since $\Paths{r}{q} \neq
  \emptyset$, there is a run $\rho_3$ from $r$ to $q$.

  We only show~\ref{l:f} as~\ref{l:e} is symmetric. We assume that $a
  < \sup \tau$, as otherwise we are done by
  \autoref{prop:obs}\ref{itm:a=sup}.
  We make a case distinction on
  whether $\sup \tau = \infty$.

  \emph{Case $\sup \tau \neq \infty$}.\ We must show that we can reach
  values arbitrarily close to $\sup \tau$, \ie\ that for every
  $\epsilon \in (0, 1]$, there exists a value $b \in [\sup \tau
  - \epsilon, \sup \tau)$ and a run $p(a) \steps{\rho} q(b)$.
      
  By \autoref{claim:scaling-down-gives-run} and $a < \sup \tau$, we
  have $p(a) \steps{(\nicefrac{1}{2}) \rho_1} r(a')$ for some $a' < \sup \tau$. Let:
  \[
  m \defeq \sum_{i=1}^{|\pi|} |\Effect{\pi_i}|,
  \alpha_t \defeq \frac{\epsilon}{4 |\Effect{t}|} \text{ and }
  \alpha_\pi \defeq \frac{\epsilon}{4m+1}.    
  \]

  Let $\rho_2 \defeq \alpha_t t\, \alpha_\pi \pi$. By definition, we
  have $\Effect{\alpha_t t} = \epsilon / 4$ and
  $|\Effect{\FromTo{\alpha_\pi \pi}{1}{i}}| < \epsilon/4$ for all $i
  \in \set{1, \ldots, |\pi|}$. Consequently, it is the case that
  $\Effect{\FromTo{\rho_2}{1}{i}} \in (0, \epsilon/2)$ for all $i \in
  \set{2, \ldots, |\rho_2|}$.

  Hence, there exists $k \geq 0$ such that $\rho_2^k$ is admissible
  from $a'$ and $\sup \tau - \epsilon/2 \leq a' + \Effect{\rho_2^k} <
  \sup \tau$. Therefore, we have:
  \[
  r(a') \steps{\rho_2^k} r(b')
  \text{ where } b' \in [\sup \tau - \epsilon/2, \sup \tau).
  \]
  By \autoref{prop:obs}\ref{claim:scale}, we can scale the run
  $\rho_3$ so that it is admissible from $r(b')$ and reaches a value
  arbitrarily close to $b'$ in state $q$. More formally, there exists
  $\beta \in (0, 1]$ such that
  \[
  r(b') \steps{\beta \rho_3} q(b)
  \text{ where } b \in [b' - \epsilon/2, \sup \tau).
  \]
  We are done since $p(a) \steps{(\nicefrac{1}{2}) \rho_1} r(a') \steps{\rho_2^k}
  r(b') \steps{\beta \rho_3} q(b)$ and $b \in [\sup \tau - \epsilon,
    \sup \tau)$.\medskip

  \emph{Case $\sup \tau = \infty$}.\ We must show that we can reach
  arbitrarily large values. Let $b \geq a$. For all $\ell \geq 0$, the
  run $\rho_\ell' \defeq (1/2) \rho_1\, \rho_2^\ell$ is admissible from
  $a$, and such that $\Effect{\rho_\ell'} > 0$. Thus, there exists
  $\ell \geq 0$ such that $\Effect{\rho_\ell'} \geq (b - a) +
  \Effectn{\rho_3}$. We are done since
  \[
  a \steps{\rho_\ell'} b' \steps{\rho_3} b''
  \text{ where } b' \geq b + \Effectn{\rho_3}
  \text{ and } b'' \geq b.\qedhere
  \]
\end{proof}

We show the characterisation of ``$a \in \Post{p,q}{a}$'' stated, but
left unproven, within the proof of \autoref{prop:a_nc2}:

\begin{proposition}
  Let $\Post{p,q}{a} \neq \emptyset$. It is the case that $a \in
  \Post{p,q}{a}$ iff at least one of these conditions~holds:
  \begin{enumerate}[label=(\alph*)]
  \item\label{a:1:app} there exists a path $\pi \in \Paths{p}{q}$ whose
    transitions are all zero, \ie\ $\Effect{\pi} = \Effectp{\pi} =
    \Effectn{\pi} = 0$;
    
  \item\label{a:2:app} there exist $\pi \in \Paths{p}{q}$ and indices $i,
    j$ such that $\Effect{\pi_i} > 0$ and $\Effect{\pi_j} < 0$. If $a
    = \inf \tau$, then we also require that $\Effect{\pi_k} = 0$ for
    all $k < i$ and $k > j$. Similarly, if $a = \sup \tau$, then we
    also require $\Effect{\pi_k} = 0$ for all $k < j$ and $k > i$.
  \end{enumerate}
\end{proposition}  

\begin{proof}
  $\Leftarrow$) If~\ref{a:1:app} holds, then trivially $a \in
  \Post{p,q}{a}$. Assume~\ref{a:2:app} holds. Let $\rho \defeq 1 t_1
  \cdots 1 t_n$ where $\pi = t_1 \cdots t_n$. Suppose $a \notin \{\inf
  \tau, \sup \tau\}$. By \autoref{prop:obs}\ref{claim:scale}, for all
  $\beta$ small enough, it is the case that $a \steps{\beta \rho}
  a_\beta$, where $|a - a_\beta| < 1/2$. Let $p(a) = q_0(a_0), \ldots,
  q_n(a_n) = q(a_\beta)$ be the sequence of configurations witnessing
  $a \steps{\beta \rho} a_\beta$. Since $n$ is fixed we can choose
  $\beta < 1/2$ small enough so that $|a_i - a| < 1/2$ for all $i$. If
  $a_\beta > a$, then we enlarge the coefficient of $t_j$ to $\alpha_j
  > \beta$ so that $(\alpha_j - \beta) \cdot \Effect{t_j} = a -
  a_\beta$. By the choice of $\beta$, we get an admissible run $\rho'
  \defeq \beta t_1 \ldots \beta t_{j-1} \alpha_j t_j \beta t_{j+1}
  \cdots \beta t_n$ that satisfies $a \steps{\rho'} a$. If $a_\beta <
  a$, then we proceed analogously with index $i$.

  It remains to prove the case where $a = \inf \tau$; the case where
  $a = \sup \tau$ is symmetric. By assumption, we have $\effect{t_k} =
  0$ for all $k < i$ and $k > j$. For the sake of simplicity, assume
  $\Effect{t_1} > 0$ and $\Effect{t_n} < 0$. Let $\alpha_1 \in (0, 1)$
  be such that $\alpha_1 \cdot \Effect{t_1} < 1/2$. Let $\rho_1 \defeq
  1 t_2 \cdots 1 t_{n-1}$. By \autoref{prop:obs}\ref{claim:scale},
  there exists $\beta \in (0, 1]$ such that $\inf \tau \steps{\alpha_1 t_1
  \beta \rho_1} \delta$, where $\delta < 1$. Since $\Effect{t_n} <
  0$, there exists $\alpha_n \in (0, 1)$ such that $\alpha_n \cdot
  \Effect{t_n} = -\delta$. Thus, we have $p(\inf \tau)
  \steps{\alpha_1 t_1 \beta \rho_1 \alpha_n t_n} q(\inf \tau)$.

  $\Rightarrow$) Let $p(a) \steps{\rho} q(a)$ and $\pi \defeq
  \upath{\rho}$. Suppose~\ref{a:1:app} does not hold. If all transitions
  of $\pi$ were positive, then we would obtain the contradiction
  $p(a) \steps{\rho} q(a')$ with $a' > a$. Similarly, all
  transitions cannot be negative. For the specific case where $a =
  \inf \tau$, observe that if the first nonzero transition is
  negative, then $\rho$ cannot be admissible. Similarly, if the last
  nonzero transition is positive then $p(\inf \tau) \steps{\rho}
  q(\delta)$ for some $\delta > \inf \tau$. The reasoning for the
  case $a = \sup \tau$ is symmetric.
\end{proof}

We show the characterisation of ``$c \in \Post{p,q}{a}$'' stated, but
left unproven, within the proof of \autoref{prop:endpoints_nc2}:

\begin{proposition}
  Let $\Post{p,q}{a} \neq \emptyset$, $b \defeq \inf
  \cl{\Post{p,q}{a}}$ and $c \defeq \sup \cl{\Post{p,q}{a}}$. If $b <
  a < c$ and $c \in \tau$, then $c \in \Post{p,q}{a}$ iff there is a
  state $r$ and a path $\sigma \in \Paths{r}{q}$ that satisfy
  $\Effectp{\sigma} > 0$, $\Effectn{\sigma} = 0$ and either of the
  following:
  \begin{enumerate}[label=(\roman*)]
  \item there exists a path $\sigma' \in \Paths{p}{r}$ such that
    $|\sigma|, |\sigma'| \leq |Q|$, $\Effectn{\sigma'} = 0$ and
    $\Effectp{\pi} \geq c - a$ where $\pi \defeq \sigma'
    \sigma$;\label{itm:sup:simple:a:app}

  \item there exists a path $\sigma' \in \Paths{p}{r}$ such that
    $|\sigma|, |\sigma'| \leq |Q|$ and $\Effectp{\pi} > c - a$ where
    $\pi \defeq \sigma' \sigma$;\label{itm:sup:simple:b:app}

  \item there is a positive $(a, p, r)$-admissible cycle
    $\theta$.\label{itm:sup:cycle:app}
  \end{enumerate}
\end{proposition}

\begin{proof}
  $\Rightarrow$) Assume $c \in \Post{p,q}{a}$. There is a run $\rho$
  such that $p(a) \steps{\rho} q(c)$. Let $\rho'$ be the run obtained
  from $\rho$ by repeatedly removing a cycle $\theta$ with
  $\Effectp{\theta} = 0$, until no further possible. Let $\pi \defeq
  \upath{\rho'}$. We have $\Effectp{\pi} \geq \Effect{\rho'} \geq
  \Effect{\rho} = c - a$. Since $c > a$, there is a maximal index $i$
  such that $\Effect{\pi_i} > 0$. Let $r \defeq \In{\pi_i}$, $\sigma'
  \defeq \FromTo{\pi}{1}{i-1}$ and $\sigma \defeq
  \FromTo{\pi}{i}{|\pi|}$. Note that $\sigma' \in \Paths{p}{r}$ and
  $\sigma \in \Paths{r}{q}$. Moreover, $\Effectp{\sigma} > 0$ holds by
  maximality of $i$. It must also be the case that $\Effectn{\sigma} =
  0$. Indeed, otherwise the last nonzero transition $t$ of $\sigma$,
  and consequently of $\rho$, would be negative. Hence, this would
  contradict $c = \sup \cl{\Post{p,r}{a}}$ as we could reach values
  arbitrarily close to $c + \epsilon$ for some $\epsilon > 0$ by
  scaling $t$ arbitrarily close to zero. Observe that if
  $\Effectp{\pi} = c - a$, then $\Effect{\rho} = \Effectp{\pi} = c -
  a$ which implies $\Effectn{\rho} = \Effectn{\pi} = 0$. Therefore, if
  $|\sigma|, |\sigma'| \leq |Q|$, we have shown~\ref{itm:sup:simple:a:app}
  or~\ref{itm:sup:simple:b:app}.

  Otherwise, $\pi$ is a nonsimple path. So, by our past cycle
  elimination, $\pi$ contains a cycle $\theta$ with $\Effectp{\theta}
  > 0$. Let us reorder $\theta$ into $\theta'$ so that the first
  transition $t$ of $\theta'$ satisfies $\Effect{t} > 0$. We have $a
  \in \enab{\Paths{p}{\In{t}}}$ as state $\In{t}$ occurs on the
  original run $\rho$ that leads to state $q$. Moreover,
  $\Paths{\In{t}}{r} \neq \emptyset$ holds by maximality of $i$. Thus,
  $\theta'$ is a positive $(a, p, r)$-admissible cycle. Hence, we have
  shown that~\ref{itm:sup:cycle:app} holds.

  \medskip

  $\Leftarrow$) If~\ref{itm:sup:simple:a:app} holds, then $\Effectp{\pi} =
  c - a$ or $\Effectp{\pi} > c - a$. The latter case is subsumed
  by~\ref{itm:sup:simple:b:app}, and in the former case we are done as $a
  \steps{\pi} c$ due to $\Effectn{\pi} = 0$. If~\ref{itm:sup:cycle:app}
  holds, then since $\theta$ is a positive $(a, p, r)$-admissible
  cycle --- and hence $(a, p, q)$-admissible ---
  \autoref{lemma:supfinf}\ref{l:f} yields $\sup \cl{\Post{p,r}{a}} = c
  = \sup \tau$. Thus, there exists $\epsilon \in [0, 1]$ such that $c
  - \epsilon \in \Post{p,r}{a}$. By $\Effectp{\sigma} \geq 1$ and
  $\Effectn{\sigma} = 0$, we have \[p(a) \steps{*} r(c - \epsilon)
  \steps{\beta \sigma} q(c) \text{ where } \beta \defeq \epsilon /
  \Effectp{\sigma}.\] If \ref{itm:sup:simple:b:app} holds, then we proceed
  as follows. Recall that $b < a < c$. Therefore, $a \notin \{\inf
  \tau, \sup \tau\}$, since $\inf \tau \leq b$ and $c \leq \sup \tau$
  by definition of $b$ and $c$. Due to $a \notin \{\inf \tau, \sup
  \tau\}$, we can scale the negative transitions of $\sigma'$
  arbitrarily close to zero and scale its positive transitions so that
  either $a + \Effectp{\sigma'} - \epsilon \in \Post{p,r}{a}$ or $\sup
  \tau - \epsilon \in \Post{p,r}{a}$ for some $\epsilon \in (0,
  1]$. Since $\Effectp{\sigma} \geq c - a - \Effectp{\sigma'} + 1$,
  $\Effectp{\sigma} \geq 1$ and $\Effectn{\sigma} = 0$, we can derive
  either $c \in \Post{r,q}{a}$ or $\sup \tau \in \Post{r,q}{a}$. As
  the latter implies $c = \sup \tau$, we are done proving the claim.
\end{proof}

\lemmaReachEquivalent*

\begin{proof}
  We show the ``if'' direction first. 
  Assume that $p(a) \steps{\rho} q(b)$ for some run $\rho$.
  Without loss of generality, we assume that no configuration repeats when starting at $p(a)$ with $\rho$; otherwise, we can simply shorten $\rho$.
  Let $\sigma_1 \sigma_2 \cdots \sigma_n$ be the unique maximal decomposition of
  $\rho$ into runs such that $\eqguard(\Out{\sigma_i}) \neq \Q$
  for all $1 \leq i < n$.
  For ease of notation, let $q_i \defeq \Out{\sigma_i}$.
  It holds that $p(a) \steps{\sigma_1} q_1(a_1) \steps{\sigma_2} q_2(a_2) \cdots \steps{\sigma_n} q(b)$,
  where $\eqguard(q_i) = [a_i,a_i]$ for all $i$.
  Since $\sigma_1 \sigma_2 \cdots \sigma_n$ is the maximal
  decomposition, $\eqguard(\Out{{(\sigma_i)}_{j}}) = \Q$ holds that
  for all $j < \abs{\sigma_i}$.

  Additionally, recall that $p$ has no incoming edges and $q$ has no outgoing edges.
  Hence, the following holds:
  \begin{alignat*}{3}
    p(a) &\steps{\sigma_1}\ && q_1(a_1)
    &&\text{ in } \V_{p, q_1}, \\
    q_{i-1}(a_{i-1}) &\steps{\sigma_i}\ && q_i(a_{i})
    &&\text{ in } \V_{q_{i-1}, q_{i}} \text{ for all } 1 < i < n, \\
    q_{n-1}(a_{n-1}) &\steps{\sigma_n}\ && q(b)
    &&\text{ in } \V_{q_{n-1}, q}.
  \end{alignat*}
  We are done as $p(a) q_1(a_1) \cdots q_{n-1}(a_{n-1}) q(b)$
  is a path of $\mathcal{G}$.

  It remains to show the ``only if'' direction.
  %Let $p(a) \steps{*} q(b)$ in $\mathcal{G}$.
  There exists a %
  Suppose there is a path $p(a) q_1(a_1) \cdots q_{n-1}(a_{n-1}) q(b)$ in $\mathcal{G}$.
  Note that if $p'(a') \steps{*} q'(b')$ in $\V_{p', q'}$, then by definition we also have 
  $p'(a') \steps{*} q'(b')$ in $\V$.
  So we have
  \[
  p(a) \steps{*} q_1(a_1) \steps{*} \cdots
  \steps{*} q_{n-1}(a_{n-1}) \steps{*} q(b) \text{ in } \V.\qedhere
  \]
\end{proof}

\subsection{Missing proofs from \autoref{sec:ptime}}\label{sec:proof-cycle}

The following lemma will be useful to prove \autoref{proposition:cycle}.

\begin{lemma}\label{lemma:bijection}
  Let $B,B' \in \mathcal{C}$ be such that $B \subseteq B'$ where $B'$ is not a
  progressing extension of $B$. There is a bijection $f \colon \intervals(B)
  \to \intervals(B')$ s.t.\ $\phi_B(I) = \phi_{B'}(f(I))$ for all $I \in
  \intervals(B)$.
\end{lemma}

\begin{proof}
  Since $B \subseteq B'$, for all $I \in \intervals(B)$ there is a unique
  $f(I) \in \intervals(B')$ such that $I \subseteq f(I)$. We show that if $f$
  is not a bijection then it will contradict that $B \subseteq B'$ is not a progressing extension.
%   
%   \michael{Is the above saying that $|\intervals(B)| = |\intervals(B')|$ as otherwise we would have a contradiction? And then we're trying to show that $f$ is a surjection via three contradictions?}
%   \philip{regarding $|\intervals(B)| = |\intervals(B')|$: 
%   no, not as it is at the moment; I think one needs \ref{itm:progr:b}
%   in addition to \ref{itm:progr:a} to show that $|\intervals(B)| = |\intervals(B')|$.}
  
  First, we prove that $f$ is an injection.
  Suppose it is not and that $f(I) = f(J)$ for some $I,J \in
  \intervals(B)$. Then $\phi_{B}(I) \cup \phi_{B}(J) \subseteq
  \phi_{B'}(f(I))$. If there exists $\ell \in \phi_{B}(I) \cap \phi_{B}(J)$,
  then it must be the case that $\ell \not \in I, \ell \not \in J$ and $\ell
  \in f(I)$. This is a contradiction because the extension is progressing due to~\ref{itm:progr:c}. Otherwise, there
  is $\ell \in \phi_{B}(I) \setminus \phi_{B}(J)$. Since $\phi_{B}(I)
  \subseteq \phi_B(I) \cup \phi_{B}(J) \subseteq \phi_{B'}(f(I))$ we get a contradiction because the extension is progressing due to~\ref{itm:progr:b}.

  Now, we prove that $f$ is a surjection. Suppose it is not.
%   Then because $f$ is an injection it must be the case that $|\intervals(B)| < |\intervals(B')|$.
  Then there
  is an interval $I' \in  \intervals(B')$ such that $f^{-1}(I') = \emptyset$.
  Thus $I' \cap B = \emptyset$, which is a contradiction because of~\ref{itm:progr:a}.
  
  To conclude, we note that by definition we have $\phi_B(I) \subseteq \phi_{B'}(f(I))$ for all $I \in
  \intervals(B)$. If the inclusion is strict for some $I$, then we get a contradiction because the extension
  is progressing due to~\ref{itm:progr:b}.
\end{proof}

\findCycle*

\begin{proof}
The value of the polynomially bounded number $n$ will be determined by
the proof. By definition, we have $S_i \preceq S_{i+1}$ for all $i \in \N$. By
\autoref{lemma:progress}, there is a polynomial number of indices
$i$ (w.r.t.\ $|Q|$) such that $S_i(q) \subseteq S_{i+1}(q)$ is a progressing extension for
some $q \in Q$. Thus, there exists an index $j$ polynomial in $|Q|$ such that the
extensions $S_i(q) \subseteq S_{i+1}(q)$ are not progressing for all $j \le i
\le j+k$, where $k$ is sufficiently large (but polynomially bounded in $|Q|$).
% \michael{I think the above might deserve symbolic quantifiers, or at least not mixing prefix/postfix quantifications.}
To simplify the notation
we will assume that $j = 0$ and consider $S_0 \preceq S_1 \preceq \dots
\preceq S_{k}$.

By \autoref{lemma:bijection}, there is a bijection $f_i \colon \intervals(S_i(q))
\to \intervals(S_{i+1}(q))$ for every $q \in Q$ and $0 \le i < k$. Thus, for
every $q \in Q$, the sets $\intervals(S_1(q))$, \dots, $\intervals(S_{k}(q))$
have the same number $m_q$ of intervals. By
\autoref{lemma:boundedunion}, $m_q$ is polynomially bounded in $|Q|$. Let us
denote the intervals $I^{q,i}_1,\dots I^{q,i}_{m_q}$ where $f(I^{q,i}_j) =
I^{q,i+1}_j$ for all $0 \le i < k$, $q \in Q$ and $1 \le j \le m_q$.  By
\autoref{lemma:bijection}, $\phi_{S_{i}(q)}(I^{q,i}_j) =
\phi_{S_{i+1}(q)}(I^{q,i+1}_j)$. Since $S_i \preceq S_{i+1}$, we conclude that
$I^{q,1}_{j} \subseteq I^{q,2}_{j} \subseteq \dots \subseteq I^{q,k}_{j}$ for
all $q \in Q$ and $1 \le j \le m_q$.
% and $S_i(q) \preceq S_{i+1}(q)$ is not progressing and that .

Recall that $T$ is the set of transitions of the guarded COCA. 
Let $\act(q,i,j) \defeq \set{t \in T \mid I^{q,i}_j \cap \enab{t} \neq
\emptyset}$. Since $I^{q,i}_j \subseteq I^{q,i+1}_j$, we have $\act(q,i,j)
\subseteq \act(q,i+1,j)$, and a strict inclusion can occur at most $|T|$
times. Since it is polynomially bounded in $|Q|$ and $|T|$, we can assume that $\act(q,i,j) = \act(q,i+1,j)$ for all $q \in Q$, $0 \le i < k$, and $1 \le j \le m_q$.
  Suppose $S_{i} \neq S_{i+1}$ for every $0 \le i < k$. Then we prove that there is a
  positively or negatively expanding cycle from $S_0$.
%can find a positively or negatively expanding cycle. 
  Since $S_{i} \neq S_{i+1}$, there exists $v_k \in S_{k}(q_k)$ for some $q_k \in
  Q$ such that $v_k \not \in S_{k-1}(q_k)$. Tracing back, by definition of
  $\TheBetterPost$, we can find a sequence of transitions $t_1,\dots, t_k$,
  scalars $\alpha_1,\dots,\alpha_k \in (0,1]$, and configurations $q_0(v_0),
  \dots, q_k(v_k)$ such that: $v_i \in S_{i}(q_i)$,  $v_i \not \in
  S_{i-1}(q_i)$, and $q_{i-1}(v_{i-1}) \steps{\alpha_i t_i} q_{i}(v_i)$ for
  all $0 <  i \le k$. We show that an infix of this sequence defines a
  positively or negatively expanding cycle.

  Let $j_k$ be the unique index with $v_k \in I^{q_k,k}_{j_k}$. By
  definition, we have $v_k \not \in I^{q_k,k-1}_{j_k}$. If $v_k \ge
  \sup(I^{q_k,k-1}_{j_k})$, then we construct a positively expanding cycle,
  and otherwise we construct a negatively expanding one. We will prove only the former case; the
  latter case follows the same steps. We claim that $v_i \ge
  \sup(I^{q_i,i-1}_{j_i})$ for all $0 < i \le k$.

  Let us argue that the claim allows us to conclude. For a large enough $k$, we can find an infix $\rho
  \defeq q_a(v_a), \dots, q_b(v_b)$ such that: $a < b$, $q_a = q_b$, and $j_a = j_b$.
  Thus, $v_b \ge \sup(I^{q_b,b-1}_{j_b}) \ge \sup(I^{q_a,a}_{j_a}) \ge
  v_a$. Since $v_b \in I^{q_b,b}_{j_b}$ and $v_a \not \in I^{q_b,b}_{j_b}$, we
  obtain $\Effect{\rho} > 0$. The remaining conditions of the positively
  expanding cycle follow directly from the definition.

  It remains to prove the claim, \ie\ that $v_i \ge \sup(I^{q_i,i-1}_{j_i})$ for every $0 < i \le
  k$. We proceed by induction going from $i = k$ down to $i = 1$. The base
  case follows by assumption. For the inductive step, towards a contradiction,
  suppose that $v_{i} \le \inf (I^{q_i,i-1}_{j_i})$ and $v_{i}  \not \in
  I^{q_i,i-1}_{j_i}$. 
%\philip{typo in the last sentence? should it be $\sup (I^{q_i,i-1}_{j_i})$?}
%\philip{nevermind, the two are equivalent by the fact that $v_{i}  \not \in I^{q_i,i-1}_{j_i}$;
%maybe this should be made explicit}
  Recall that $\act(q_i,i-1,j_i) = \act(q_i,i,j_i)$. Since $v_i \steps{\alpha_{i+1} t_{i+1}} v_{i+1}$, there exists $\overline{v} \in I^{q_i,i-1}_{j_i}$ and $\beta \in (0,1]$ such that $\overline{v} \steps{\beta t_{i+1}}\overline{w}$ for some $\overline{w} \in S_{i}(q_{i+1})$.

  We show that $v_{i+1} < \overline{w}$.
  If $v_{i+1} = \overline{w}$, then $v_{i+1} \in S_{i}(q_{i+1})$, which
  contradicts the definition of $v_{i+1}$.
  Suppose that $v_{i+1} \neq \overline{w}$ for any choice of $\beta$. Let $t_{i+1} = (q_{i}, z_{i+1}, q{i+1})$. Since $v_{i} \le \inf (I^{q_i,i-1}_{j_i})$, notice that $\overline{v} + \alpha_{i+1}z_{i+1} > v_{i} + \alpha_{i+1}z_{i+1} = v_{i+1}$.
  Thus, it must be the case that $v_{i+1} < \overline{w}$.

  We prove that $\overline{w} \in I^{q_{i+1},i+1}_{j_{i+1}}$.
  Recall that $v_{i}, \overline{v} \in I^{q_i,i}_{j_i}$ and thus $[v_i,\overline{v}] \subseteq I^{q_i,i}_{j_i}$. Also $[v_i + \alpha_{i+1}z_{i+1}, \overline{v} + \beta z_{i+1}] = [v_{i+1},\overline{w}]$. Thus, for every $v_{i+1} \le w' \le \overline{w}$ there exists $v_i \le v' \le \overline{v}$ and $\gamma \in (0,1]$ such that $v' \steps{\gamma t_{i+1}} w'$. Thus, $\overline{w}$ and $v_{i+1}$ belong to the same interval in $S_{i+1}(q_{i+1})$ as required.

  Since $\overline{w} \in I^{q_{i+1},i+1}_{j_{i+1}}$ and $\overline{w} \in S_{i}(q_{i+1})$, we have $\overline{w} \in I^{q_{i+1},i}_{j_{i+1}}$. We have reached a contradiction, since by the inductive hypothesis we have: \[v_{i+1} \ge \sup(I^{q_{i+1},i}_{j_{i+1}}).\]

Observe that since $\TheBetterPost$ is easily computable in polynomial time, one can also find $q_0(v_0), \dots, q_k(v_k)$ and their corresponding intervals in polynomial time.
\end{proof}

%%%%%%%%%%%%%%%%%%%%%%%%%%%%
%% APPENDIX OF PARAMETRIC REACHABILITY %%
%%%%%%%%%%%%%%%%%%%%%%%%%%%%

\subsection{Missing formulas from \autoref{sec:parametric}}\label{app:parametric}

\paragraph*{Intersection of intervals}

Let $\chi_x$ be the auxiliary function defined by $\chi_x(y, z) \defeq
y$ if $x \neq 2$ and $z$ otherwise. Let $I = (b,t,\bot,\top) \in
\encodings$ and $I = (b',t',\bot',\top') \in \encodings$. The
intersection $I \cap I \defeq (b'', t'', \bot'', \top'')$ can be
defined by
\begin{align*}
  b''
  &\defeq
  \max(\chi_\bot(b, b'), \chi_{\bot'}(b', b)), \\
  t''
  &\defeq
  \min(\chi_\top(t, t'), \chi_{\top'}(t', t)), \\
  \bot''
  &\defeq
  \begin{cases}
    \chi_{\bot}(\bot, \bot')  & \text{if } b > b', \\
    \chi_{\bot'}(\bot', \bot) & \text{if } b < b', \\
    \min(\bot, \bot')         & \text{otherwise},
  \end{cases} \\
  \top''
  &\defeq
  \begin{cases}
    \chi_{\top}(\top, \top')  & \text{if } t < t', \\
    \chi_{\top'}(\top', \top) & \text{if } t > t', \\
    \min(\top, \top')         & \text{otherwise}.
  \end{cases}
\end{align*}

To explain the above,
let us first look at the case where both intervals are bounded,
that is, $\bot, \top, \bot', \top' \in \{0,1\}$.
Each invocation of $\chi$ will return its first argument.
Let us examine the lower endpoint.
It is the case that $b'' = \max(b,b')$. As expected, the intersection
operation uses the larger lower endpoint among the two intervals.
Further, whether the lower endpoint of the interval is included
depends on whether the value that was used was
included in the interval it originated from, \ie\ consider
$(3,5] \cap [4,6] = [4,5]$. The only remaining nuance is
when both intervals have the same endpoint,
then it is only included if it is included in both intervals.
This is why we choose $\min(\bot,\bot')$ in the case where $b = b'$.

Now, let us consider the case where either interval may be unbounded.
Then it is not sufficient to compare endpoints, \eg\
the encoding $(0,1,2,2)$ does not encode an interval with endpoints
$0$ and $1$, but rather $(-\infty, +\infty)$.
Intuitively, the $\chi$ function is used in order to
handle such unbounded intervals. When intersecting two intervals,
where one interval has no finite lower endpoint, while the other does,
the lower endpoint should only depend on the interval that indeed has
a lower endpoint. So for determining the lower endpoint and whether it is included,
the $\chi$ function will check whether the interval is
bounded, and only use its endpoint if it exists.
Otherwise, it defaults to the endpoint of the other interval.

The reasoning behind the upper endpoint is symmetric.

\subsection{Missing gadgets from the proof of \autoref{thm:acyclic-np-complete}}

Recall that the proof of \autoref{thm:acyclic-np-complete} gives the
gadgets of $\Pp$, but not of $\Pp'$. \autoref{fig:np:complete-updates}
depicts the missing gadgets of $\Pp'$.

\begin{figure}
  \begin{center}
    \begin{tikzpicture}[auto, thick, scale=0.9]
  \tikzstyle{astate} = [state, minimum size=15pt, inner sep=0pt];
  %% Guessing
  \node[astate, label={above:$[0, 0]$}, right=1cm of g4]    (h0) {$p_i$};
  \node[astate, label={above:$[0, 1]$}, right=0.75cm of h0] (h1) {};
  \node[astate, label={left:$[1, 1]$}, below=0.75cm of h1] (h2) {};
  \node[astate, label={above:$[0, 0]$}, right=0.75cm of h1] (h3) {$q_i$};

  \path[->]
  (h0) edge node       {$x_i$} (h1)
  (h1) edge node       {$0$}   (h3)
  (h1) edge node[swap] {$0$}   (h2)
  (h2) edge node[swap] {$-1$}  (h3)
  ;
  
  %% Checking
  \node[astate, label={above:$[0, 0]$}, right=0.5cm  of c4] (d0) {$r_j$};
  \node[astate, label={above:$[1, 1]$}, right=1cm    of d0] (d2) {};
  \node[astate, label={above:$[1, 1]$}, above=0.75cm of d2] (d1) {};
  \node[astate, label={above:$[0, 0]$}, below=0.75cm of d2] (d3) {};
  \node[astate, label={above:$[0, 0]$}, right=1cm    of d2] (d4) {$s_j$};

  \path[->]
  (d0) edge node       {$x_1$} (d1)
  (d0) edge node       {$x_2$} (d2)
  (d0) edge node[swap] {$x_4$} (d3)
  (d1) edge node[above left, xshift=10pt, yshift=+5pt] {$-1$}  (d4)
  (d2) edge node       {$-1$}  (d4)
  (d3) edge node[below left, xshift=10pt, yshift=-5pt] {$0$}  (d4)
  ;
\end{tikzpicture}
  \end{center}
  \caption{Gadgets used in the reduction from 3-SAT to parametric COCA
    reachability where parameters occur only on updates. 
    \emph{Left}: Gadget for variable $x_i$;
    \emph{Right}: Gadget used for clause $C_j = (x_1 \lor x_2 \lor \neg
    x_4)$.}
  \label{fig:np:complete-updates}
\end{figure}
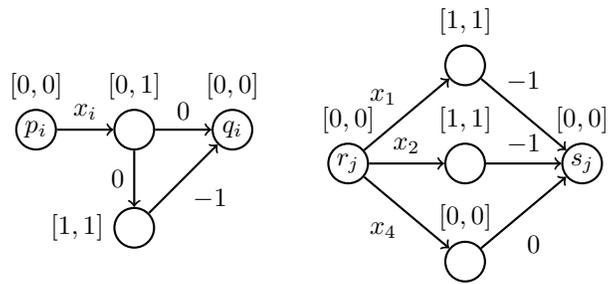

\end{document}